\documentclass[conference]{IEEEtran}
\IEEEoverridecommandlockouts

\usepackage{amsmath, amssymb, amsthm}
\usepackage{graphicx}
\usepackage{booktabs}
\usepackage{nicefrac}
\usepackage{siunitx}
\usepackage{subcaption}
\usepackage{balance}
\usepackage{multirow}
\usepackage{tabularx}
\usepackage{makecell}
\usepackage{enumitem}
\usepackage{xcolor}
\usepackage{xurl}
\usepackage{pifont}
\usepackage{textcomp}
\usepackage{xspace}

\usepackage[linesnumbered,ruled,vlined,norelsize]{algorithm2e}

\usepackage{hyperref}

\usepackage{fancyhdr}
\newtheorem{definition}{Definition}
\newtheorem{theorem}{Theorem}
\newtheorem{lemma}{Lemma}

\def\BibTeX{{\rm B\kern-.05em{\sc i\kern-.025em b}\kern-.08em
    T\kern-.1667em\lower.7ex\hbox{E}\kern-.125emX}}

\begin{document}

\DeclareRobustCommand{\hym}[1]{#1}
\DeclareRobustCommand{\ym}[1]{#1}
\DeclareRobustCommand{\xht}[1]{#1}
\newcommand{\tj}[1]{#1}
\newcommand{\tjr}[1]{#1}
\newcommand{\tbd}[1]{#1}
\DeclareRobustCommand{\rTwo}[1]{#1}
\DeclareRobustCommand{\rThree}[1]{#1}
\DeclareRobustCommand{\rFour}[1]{#1}
\DeclareRobustCommand{\rFive}[1]{#1}
\DeclareRobustCommand{\meta}[1]{#1}

\newcommand{\jing}[1]{}
\newcommand{\unsure}[1]{}
\newcommand{\delete}[1]{} 
\newcommand{\tjq}{}
\newcommand{\pending}{}

\newcommand{\revise}[1]{#1}

\newcommand{\eat}[1]{\ignorespaces}  
\newcommand{\spara}[1]{\vspace{1.5mm}\noindent\textbf{#1.}}
\newcommand{\originalcomment}[1]{\textit{#1}}

\newcommand{\compilehidecomments}{true}

\title{PRIME: Efficient Algorithm for Token Graph Routing Problem\thanks{A short version of the paper will appear in the 42nd IEEE International Conference on Data Engineering (ICDE '26), May 4--8, 2026.}}

\makeatletter
\newcommand{\linebreakand}{%
  \end{@IEEEauthorhalign}
  \hfill\mbox{}\par
  \mbox{}\hfill\begin{@IEEEauthorhalign}
}
\makeatother
\author{
    \IEEEauthorblockN{Haotian Xu}
    \IEEEauthorblockA{The Hong Kong University of Science and Technology (Guangzhou)\\ hxu496@connect.hkust-gz.edu.cn}
    \and
    \IEEEauthorblockN{Yuqing Zhu}
    \IEEEauthorblockA{Nanyang Technological University\\ yuqing002@e.ntu.edu.sg}
    \linebreakand
    \IEEEauthorblockN{Yuming Huang}
    \IEEEauthorblockA{National University of Singapore\\ huangyuming@u.nus.edu}
    \and
    \IEEEauthorblockN{Jing Tang\IEEEauthorrefmark{1}\thanks{\IEEEauthorrefmark{1}Corresponding author: Jing Tang.}}
    \IEEEauthorblockA{The Hong Kong University of Science and Technology (Guangzhou)\\ jingtang@ust.hk}
}

\maketitle
\thispagestyle{fancy} 

\begin{abstract}
Optimizing asset exchanges on blockchain-driven platforms poses a novel and challenging graph query optimization problem. \meta{In this model, assets represent vertices and exchanges form edges, recasting the graph query task as a routing problem over a large-scale, dynamic graph. However, the existing solutions fail to solve the problem efficiently due to the non-linear nature of the edge weights defined by a concave swap function.} To address the challenge, we propose PRIME, a two-stage iterative graph algorithm designed for the Token Graph Routing Problem (TGRP). The first stage employs a pruned graph search to efficiently identify a set of high-potential routing paths. The second stage formulates the allocation task as a strongly convex optimization problem, which we solve using our novel Adaptive Sign Gradient Method (ASGM) \rFour{with} a linear convergence rate. Extensive experiments on real-world Ethereum data confirm PRIME's advantages over industry baselines. PRIME consistently outperforms the widely-used Uniswap routing algorithm, achieving up to 8.42 basis points (bps) better execution prices on large trades while reducing computation up to 96.7\%. The practicality of PRIME is further validated by its deployment in hedge fund production environments, demonstrating its viability as a \meta{scalable graph query processing} solution for high-frequency decentralized markets.
\end{abstract}


\section{Introduction}
\label{sec:intro}
\subsection{Background}
\label{subsec:background}
Blockchain technology has forged a paradigm shift in financial infrastructure, culminating in a digital asset market exceeding $3.7$ trillion as of September 2025. Central to this ecosystem are decentralized exchanges (DEXs), which are programs operated on the blockchain and now captures over $20\%$ of global trading activity~\cite{milkroaddaily2025}. However, the rapid evolution of these protocols leads to the dispersion of liquidity across different trading pairs~\cite{lehar2023liquidity}. In this fragmented liquidity landscape, a fundamental challenge arises: given a source asset and a target asset, how to route funds through a network of intermediate liquidity pools to maximize the output amount? This challenge is formalized as the Token Graph Routing Problem (TGRP). With daily volumes on platforms like Ethereum exceeding billions of dollars~\cite{defillamadex}, efficient TGRP solutions are not merely algorithmic improvements but direct drivers of economic efficiency.

\meta{
Fundamentally, TGRP extends the Generalized Network Flow (GNF) problem~\cite{ahuja1993network}. Unlike standard network flow where the flow amount remains constant across edges, GNF assigns a gain factor $\gamma(e)$ to each edge, representing the exchange rate. However, TGRP introduces a unique property that distinguishes it from classical GNF: concavity. While classical GNF algorithms assume constant gain factors, TGRP's edge weights are defined by concave functions. In these settings, the effective gain is not static but dynamically decreases as input flow increases—a phenomenon known as price slippage. Consequently, TGRP constitutes a Convex Generalized Network Flow Maximization problem. This convexity renders traditional linear programming or combinatorial GNF algorithms inapplicable, necessitating a novel algorithmic approach capable of handling concave edge weights efficiently.
}

\rThree{
Beyond theoretical complexity, practical TGRP optimization faces severe structural and numerical constraints that defy existing query processing techniques. The primary hurdle is the challenge of scalability, as the graph scale is massive; for instance, a single DEX like Uniswap contains over 400,000 unique tokens and 380,000 liquidity pools~\cite{uniswapv2core}. \rTwo{Second, the ecosystem suffers from extreme numerical heterogeneity, where asset prices can differ by orders of magnitude (e.g., BTC at $10^{5}$ USD vs. SHIB at $10^{-6}$ USD~\cite{coingecko})}. This disparity leads to significant numerical instability in traditional optimization methods~\cite{github_issue_29}. Finally, these complexities must be resolved within strict latency bounds. Therefore, as market states change rapidly, these limitations make the problem a challenging graph query optimization problem and existing approaches are ill-equipped.
}

\subsection{Contribution}

To address the aforementioned challenges, we introduce PRIME, a novel algorithmic framework designed to solve the Token Graph Routing Problem (TGRP) in real-world DEX environments. To overcome the scalability challenge, we leverage the observation that liquidity in real-world markets often exhibits a clustering effect~\cite{gabaix2009power}, where a small subset of assets captures the majority of trading volume. Consequently, centering the search around high-liquidity assets yields high efficiency. PRIME pre-processes the massive token graph into two components: a Core Graph and a Shortcut Index. The Core Graph includes the primary assets and captures the main routing backbone. The Shortcut Index is utilized to record superior pricing paths that traverse non-core tokens. This decomposition allows PRIME to capture the most favorable prices across the entire ecosystem without the overhead of an exhaustive search.

\rTwo{
To tackle the computational complexity introduced by non-linear, non-additive concave swap functions, we design an iterative path discovery mechanism motivated by classical primal-dual methods for Generalized Network Flows~\cite{truemper1973optimal}. Our approaches dynamically
constructs the solution set by iteratively identifying ``generalized augmenting paths''---routes that offer superior marginal price compared to the current price threshold. We implement this via a specialized Breadth-First Search Pruning algorithm that filters edges based on dynamic liquidity thresholds. This method efficiently narrows the search space during exploration, achieving a time complexity of $O(|V| \cdot |E|)$, where $|V|$ and $|E|$ are the number of vertices (tokens) and edges.

To mitigate the numerical instability caused by extreme value differences, we propose an Adaptive Sign Gradient Method (ASGM). ASGM is a gradient-based optimization method that is robust to extreme value differences. It uses only the sign of the gradient for determining the reallocation direction and adaptively adjusts the step size using the Armijo condition~\cite{armijo1966minimization}. This design ensures robustness against extreme price disparities that destabilize traditional solvers~\cite{github_issue_29}. Furthermore, we prove that our allocation problem is strongly concave, allowing ASGM to guarantee a linear convergence rate. Collectively, these innovations enable PRIME to efficiently solve the TGRP in dynamic, real-world DEX environments.

}

To validate PRIME's effectiveness, we conduct extensive experiments using real-world Ethereum data against industry baselines like the Uniswap Smart Order Router (SOR~\cite{uniswap_routing}) in different market periods. Additionally, we compare PRIME against a theoretical optimum derived from a global convex solver, confirming that PRIME produces near-optimal outputs with negligible deviation. In summary, our contributions are as follows:

\begin{enumerate}[topsep=2mm, partopsep=0pt, itemsep=1mm, leftmargin=18pt]
    \item We formalize the Token Graph Routing Problem (TGRP) as a novel graph query optimization problem where path values are the composition of non-linear concave functions. This formulation accurately aligns with real-world token swap transactions in the blockchain. 
    \item We introduce PRIME, a novel iterative framework featuring a graph pre-processing technique, an efficient path discovery algorithm (a BFS-based pruning method with $O(|V| \cdot |E|)$ complexity) and a robust optimization algorithm with a linear convergence rate.
    \item We perform extensive evaluations on real-world datasets. The results demonstrate that PRIME achieves superior scalability and efficiency compared to classical graph algorithms and industry heuristics, offering a viable solution for high-frequency decentralized markets.
\end{enumerate}

\eat{
The core economic principle under our approach is that the optimal output is achieved when the marginal prices across all utilized paths are equalized. At this equilibrium, no further gains can be made by reallocating input. PRIME decomposes the problem into two manageable sub-problems: first, efficiently discovering high-potential paths, and second, developing a robust method to optimally allocate funds across these chosen paths. By iteratively seeking and incorporating improvements indicated by marginal price signals, PRIME continuously refines the trading solution to maximize the final output amount.

For path discovery, the challenge lies in the graph's unique properties where path values are determined by the composition of non-linear, non-additive concave functions. This structure violates the optimal substructure assumption required by traditional algorithms like Dijkstra's. To address this problem, we design a Breadth-First Search Pruning algorithm to discover optimal paths while reducing computational complexity. It incorporates a pruning technique to significantly narrow the search space during the exploration process, achieving a time complexity of $O(|V| \cdot |E|)$, where $|V|$ and $|E|$ are the number of vertices (tokens) and edges. This approach avoids the computational burden of exhaustive searches while identifying superior paths missed by simple heuristics.

For allocation optimization, while the general routing problem is convex, global solutions based on dual methods\cite{diamandis2023efficient} face two critical limitations in practice. First, their convergence is merely sublinear, making them too slow for time-sensitive DEX markets. Second, they can be numerically unstable when token prices differ greatly (e.g., differences exceeding $10^{11}$ between BTC and meme coins like SHIB~\cite{coingecko}), making them not robust~\cite{github_issue_29}. \eat{We solve this problem by focusing on a path-constrained allocation sub-problem. we prove that the sub-problem is strongly concave and the problem can be efficiently solved with a guaranteed linear convergence rate.} To overcome these limitations, we introduce the Adaptive Sign Gradient Method (ASGM) for the allocation sub-problem. We prove this sub-problem is strongly concave, which allows ASGM to achieve a guaranteed linear convergence rate, ensuring both computational speed and numerical robustness against extreme price differences.

To validate PRIME's effectiveness and efficiency, we conduct extensive experiments using real-world Ethereum data, comparing it against the widely-used Uniswap Smart Order Router (SOR). Our results demonstrate PRIME’s superiority in both execution price and computational efficiency, achieving up to 28 basis points (bps) better prices on large trades while reducing computation time by 35\% to 96\% across all test cases. These performance gains remained stable across diverse market conditions. Further, we also benchmark PRIME against a convex optimization method, which provides the theoretical optimum. The results confirm that PRIME produces near-optimal outputs, with only a negligible difference from the theoretical best. The practicality of PRIME is further validated by its successful deployment in hedge funds. In summary, our main contributions are as follows.

\begin{enumerate}[topsep=2mm, partopsep=0pt, itemsep=1mm, leftmargin=18pt]
    \item We formalize the Token Graph Routing Problem (TGRP) as a novel graph query optimization problem where path values are the composition of non-linear concave functions. This formulation accurately aligns with real-world token swap transactions in the blockchain.  
    
    \item We introduce PRIME, a novel iterative framework featuring an efficient path discovery algorithm (a BFS-based pruning method with $O(|V| \cdot |E|)$ complexity) and a robust optimization algorithm with a linear convergence rate.
    
    \item We carry out extensive experiments on three market periods to evaluate the effectiveness of PRIME in real-world DEXs. Experimental results confirm that our algorithm can improve execution price while reducing computational time compared to classical graph algorithm and current industry algorithms. Besides, we benchmark against a convex optimization method, revealing that PRIME can efficiently produce near-optimal output.

\end{enumerate}
}

\eat{For allocation optimization, we invent the Adaptive Sign Gradient Method (ASGM), a robust method to calculate the optimal allocations of input. The primary challenge in token swap optimization arises from extreme value disparities between tokens (e.g., differences exceeding $10^{11}$ between BTC and meme coins like SHIB~\cite{coingecko}). This leads to vanishingly small gradient values for some token pair. Consequently, traditional optimization methods, such as the Projected Gradient Method (PGM)~\cite{calamai1987projected}, are unsuitable for this problem because they rely on precise gradient values and it can lead to numerical instability when solving the problem. To address the challenge, ASGM innovatively uses only the sign of the gradient for determining the reallocation direction and adaptively adjusts the step size using the Armijo condition~\cite{armijo1966minimization}. This makes ASGM robust to extreme value differences. Furthermore, we also provide a theoretical proof demonstrating that ASGM converges to an optimal allocation with a guaranteed sub-linear convergence rate of at least $O(1/T) $, where $ T $ is the number of iterations. This means that the optimality gap decreases inversely with the number of iterations, ensuring fast, stable convergence.

To validate PRIME's effectiveness and efficiency, we conduct extensive experiments using real-world Ethereum data, comparing it against the widely-used Uniswap Smart Order Router (SOR)~\cite{uniswap_routing}. Our results demonstrate PRIME’s superiority in both execution price and computational efficiency, achieving up to 7 basis points (bps) better prices on small trades and striking improvements of over 300 bps on large trades for assets like WBTC and LINK. In terms of computational efficiency, PRIME reduces computation time by 24\% to 96\% across all test cases. To demonstrate the stability of these performance gains, we tested PRIME across diverse market dynamics, using extensive data from bearish, stable, and bullish periods. Our results confirm that PRIME's superiority in both price and speed is consistently maintained regardless of market volatility. To further asses the quality of our solution, we also benchmark PRIME against a convex optimization method, which provides the theoretical optimum. The results confirm that PRIME produces near-optimal outputs, with only a negligible difference from the theoretical best. The practicality of our approach is further reinforced by demonstrating our model's consistency with major DeFi protocols and its successful deployment in hedge funds. In summary, our main contributions are as follows.

\begin{enumerate}[topsep=2mm, partopsep=0pt, itemsep=1mm, leftmargin=18pt]
    \item We formalize the Token Graph Routing Problem (TGRP) as a novel graph query optimization problem where path values are the composition of non-linear concave functions, a structure that invalidates classical shortes path algorithms. This formulation accurately aligns with real-world token swap transactions in the blockchain.  
    
    \item We introduce PRIME, a novel iterative algorithm including an efficient path discovery method and a robust allocation optimization method via the Adaptive Sign Gradient Method (ASGM). The optimization algorithm is designed to solve the numerical instability caused by extreme value disparities in token markets, and we provide a rigorous proof of its convergence property.
    
    \item We carry out extensive experiments on three market periods to evaluate the effectiveness of PRIME in real-world DEXs. Experimental results confirm that our algorithm can improve execution price while reducing computational time compared to current industry algorithms. Besides, we benchmark against a convex optimization method, revealing that PRIME can efficiently produce near-optimal output.
\end{enumerate}

}

\subsection{Organization}
The remainder of this paper is organized as follows. Section~\ref{sec:background} provides background on the decentralized financial market. Section ~\ref{sec:formulation} formalizes the Token Graph Routing Problem (TGRP). Section~\ref{sec:existing} surveys existing solutions and analyzes their limitations. Section ~\ref{sec:algorithm} details the proposed PRIME algorithm. Section ~\ref{sec:experiment} presents an extensive empirical evaluation using real-world Ethereum data. Section ~\ref{sec:Discussion} discusses practical industry implementation issues. Section~\ref{sec:relate} reviews related work, and Section~\ref{sec:conclude} concludes the paper.
\section{Decentralized Financial Market}
\label{sec:background}

\spara{Blockchain and Decentralized Finance (DeFi)}
Blockchain technology is a distributed ledger that enables secure, transparent, and tamper-resistant transaction recording across a distributed network without centralized intermediaries~\cite{nakamoto2008bitcoin}. Platforms like Ethereum expanded this infrastructure through smart contracts~\cite{wood2014ethereum}. Smart contracts are self-executing pieces of code deployed on the blockchain that automatically enforce the terms of an agreement when predefined conditions are met~\cite{szabo1997formalizing}. These technologies underpin Decentralized Finance (DeFi), a paradigm shift in financial systems. The new paradigm enables users to trade, lend, and stake assets peer-to-peer without reliance on central authorities. DeFi supports diverse applications, including decentralized exchanges~\cite{schar2021decentralized}, lending platforms~\cite{gudgeon2020defi}, and yield farming strategies~\cite{werner2022sok}.

\spara{Decentralized Exchange (DEX)}
Decentralized exchanges (DEXs) are core pillars of the DeFi ecosystem, enabling peer-to-peer trading without centralized custody~\cite{schar2021decentralized}. Unlike traditional exchanges that rely on order books, most DEXs utilize Liquidity Pool smart contracts. By maintaining dual-sided asset reserves, these protocols allow users to trade directly against the contract, where the execution price is dynamically adjusted according to the resulting asset imbalance in the pool. To facilitate this process, liquidity providers supply these tokens to the pools in exchange for a share of transaction fees. Currently, these pools operate on a wide array of digital assets, such as Wrapped Ether (WETH), an ERC-20 compliant version of the Ethereum native token; Wrapped Bitcoin (WBTC), a tokenized representation of Bitcoin on Ethereum; and stablecoins like USDC and USDT which are pegged to the U.S. dollar. This automated, on-chain mechanism forms the foundation for decentralized trading, offering a more transparent and efficient financial system.

\eat{
\spara{Blockchain and Smart Contracts} Blockchain technology \rFour{is} a distributed ledger that enables secure, transparent, and tamper-resistant transaction recording across a distributed network without centralized intermediaries~\cite{nakamoto2008bitcoin}. The evolution of blockchain platforms, notably Ethereum, extended its capabilities by introducing smart contracts~\cite{wood2014ethereum}. Smart contracts are self-executing pieces of code deployed on the blockchain that automatically enforce the terms of an agreement when predefined conditions are met~\cite{szabo1997formalizing}. Consequently, smart contracts have become a foundational technology for the expansion of decentralized ecosystems, fostering innovation and new business models like Decentralized Finance.

\spara{Decentralized Finance}
Decentralized Finance (DeFi) represents a paradigm shift in financial systems, leveraging blockchain and smart contracts to offer alternatives to traditional intermediaries like banks and brokers. Ethereum plays a pivotal role in DeFi's emergence. It enables users to trade, lend, borrow, and stake assets directly on-chain\cite{aquilina2024decentralized}. The core innovation of the paradigm shift lies in the use of smart contracts that automate and enforce agreements, enabling peer-to-peer interactions without reliance on central authorities. This enables a range of \rFour{applications} including decentralized exchanges\cite{schar2021decentralized}, lending platforms\cite{gudgeon2020defi}, and yield farming strategies~\cite{werner2022sok}. These protocols operate on a wide array of digital assets, such as Wrapped Ether (WETH), an ERC-20 compliant version of Ethereum native token; Wrapped Bitcoin (WBTC), a tokenized representation of Bitcoin on Ethereum; and stablecoins like USDC and USDT which are pegged to the U.S. dollar. Ultimately, DeFi offers the potential for a more transparent and efficient financial system.

\spara{Decentralized Exchange (DEX)}
Decentralized exchanges (DEXs) are core pillars of the DeFi ecosystem, enabling peer-to-peer cryptocurrency trading without centralized custody. Unlike traditional exchanges that order book to match buyers and sellers, most DEX are built on \textbf{Liquidity Pools} smart contracts, which hold reserves of two or more tokens, enabling the direct swaps between these tokens. Users known as liquidity providers supply tokens to these pools in exchange for a share of the transaction fees generated from trades. This automated, on-chain liquidity mechanism forms the foundation for decentralized trading.

}

\spara{Constant Function Market Maker (CFMM)}
At the core of each liquidity pool lies the Constant Function Market Maker (CFMM), an algorithm that deterministically sets asset prices. A CFMM is an algorithm that determines the price of assets based on a predefined mathematical formula in the pool. For example, Uniswap V2~\cite{uniswapv2core} employs the constant product model. Specifically, it maintains the invariant $x \cdot y= K$, where $x$ and $y$ denote the reserve quantities of Token A and Token B within the liquidity pool, respectively. And $K$ is a constant. If a user inputs an amount $\Delta x$ of token A to swap for Token B, the pool calculates the output amount  $\Delta y$ of Token B returned to the user such that the invariant is preserved  post-trade: $(x+\Delta x) \cdot (y-\Delta y)= K$. Solving for the output amount yields $\Delta y = y - \frac{K}{x+\Delta x}$. Following this model, subsequent protocols introduced more sophisticated functions to enhance capital efficiency. For instance, Uniswap V3~\cite{adams2021uniswap} implements concentrated liquidity that enables liquidity providers to specify price ranges and optimize capital efficiency, while protocols like Curve~\cite{curve_egorov2019stableswap} and Balancer~\cite{balancer} offer specialized functions for stablecoins and multi-asset pools, respectively.

\eat{
Liquidity Pool Contracts are a class of smart contracts that hold reserves of two or more tokens and facilitate trading between them~\cite{uniswapv2core,adams2021uniswap,balancer,curve_egorov2019stableswap}. \eat{Liquidity providers deposit token pairs into these pools, and in return, they earn transaction fees generated from swaps. } The price of tokens in a liquidity pool is algorithmically determined by the pool’s reserves and a predefined mathematical function. For example, Uniswap V2~\cite{uniswapv2core} adopts the equation $x*y= K$, where $x$ and $y$ are the number of token A and token B locked in the liquidity pool and $K$ is constants. If a user wants to use $\Delta x$ amount of token A to swap token B, the output amount can be calculated as $\Delta y = y - \frac{K}{x+\Delta x}$. Later iterations of DEX protocols introduced more sophisticated constant function designs to enhance asset efficiency. For instance, Uniswap V3 ~\cite{adams2021uniswap} implemented concentrated liquidity, enabling liquidity providers to specify price ranges and optimize capital efficiency. Curve employs a hybrid constant function for token swaps between similarly priced assets~\cite{curve_egorov2019stableswap}. Balancer introduced the concept of customizable token weights within a liquidity pool for multiple assets~\cite{balancer}. This flexibility allows DEX to design more suitable pool for different asset pairs. Currently, Uniswap remains the dominant DEX on Ethereum by TVL, while Curve dominates stable-coin trading volume.}

\spara{Example of a Token Graph Route}
To maximize the amount of tokens a user receives, a token graph route often involves splitting the input amount across multiple trading paths. We illustrate this with an example in Figure \ref{fig:Token Swap Transaction}, where a user aims to swap a total amount of WETH for the maximum possible amount of USDC. The figure models the DeFi market as a token graph: nodes represent digital assets (e.g., WETH, USDC, etc.), and edges represent liquidity pools enabling swaps between them. The selected route, highlighted by solid lines, demonstrates a split-routing strategy. Specifically, the total input is divided across three distinct routes: 30\% is routed through the direct path WETH → USDC, 50\% is channeled through the multi-hop path WETH → USDT → USDC, and the remaining 20\% traverses a second multi-hop path via WETH → WBTC → USDC. The final amount of USDC received by the user is the sum of the outputs from these three parallel paths. This example demonstrates that maximizing trade output requires both discovering the best combination of paths and determining the optimal allocation of funds.

\begin{figure}[!t]
    \centering
    \includegraphics[width=0.98\linewidth]{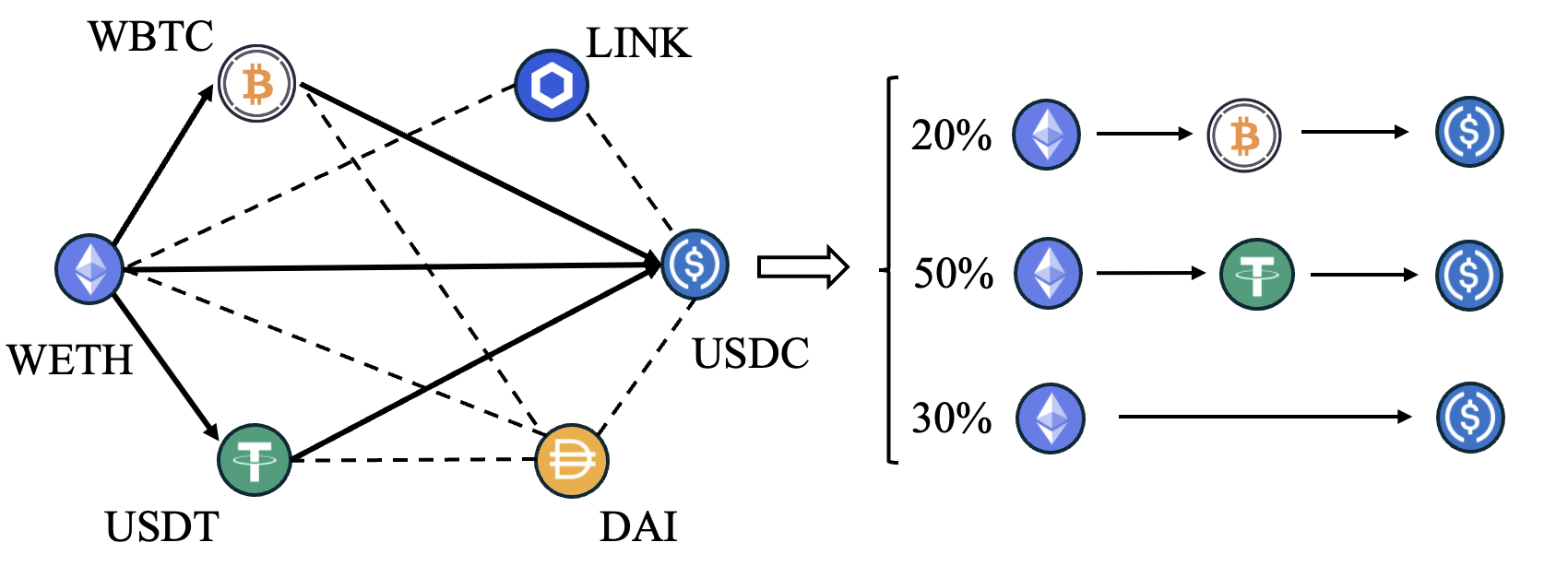}
    \caption{An Example Token Route from WETH to USDC}
    \label{fig:Token Swap Transaction}
\end{figure}

\section{Problem Formulation}
\label{sec:formulation}
\eat{Prior work has modeled token swaps on blockchain, but their algorithms often do not address real-world constraints. Therefore, we redefine the \textbf{Token Swap Routing Problem} to align with actual execution mechanisms.}

\eat{In this section, we define the that is consistent with the practical token swap transactions in DEX. First, we define key concepts essential to our models and then formally formulate the problem. For clarity, we summarize the notations frequently used in our model in Table~.} 
In this section, we formalize the \textbf{Token Graph Routing Problem (TGRP)}, ensuring alignment with real-world token swap mechanics in DEXs. We first introduce the essential concepts and system models, followed by the mathematical formulation of the problem. For clarity, Table \ref{tab: notations} summarizes the frequently used notations throughout the paper.

\eat{Prior work has modeled token swaps on blockchain, but their models do not conform to reality. Therefore, we redefine the \textbf{Token Swap Routing Problem} to align with actual execution mechanisms. In the preceding section, we introduced the practical execution mechanism of token swaps on blockchain and elucidated the previous related theoretical research. While these research conducted a great model to formalize the problem, the algorithm they developed failed to address the problem under real-world constraints. Therefore, we redefine the \textbf{Token Swap Routing Problem} to be consistent with the actual token swap execution mechanism. \hym{the above is a bit wordy} }\eat{In the practical mechanism, it will restrict token swaps via paths directed from the source token to the target token\hym{i still feel the logic a bit weird but still ok for me.}. Specifically, this restriction ensures that swaps follow a sequential, directed sequence of tokens. For instance, to swap from Token A to Token C, the swap must proceed through a defined path, potentially involving an intermediate token like Token B (i.e., Token A -> Token B -> Token C). \hym{unknown why need this example, what you want to say, seams no need.}}



To begin with, a swap function describes the mathematical relationship between the input and output amount of asset within a decentralized exchange (DEX) liquidity pool. The formal definition is as follows: 

\begin{definition}[Swap Function]
\label{def: swap function}
A swap function \( f: \mathbb{R}_{\geq 0} \rightarrow \mathbb{R}_{\geq 0} \) is a strictly concave, monotonically increasing function satisfying \( f(0) = 0 \), mapping an input token amount \( x \) to an output token amount \( f(x) \).
\end{definition}

\eat{This swap function has the following properties. Firstly, the zero origin property, $f(0)=0$, indicates that if the input token amount is zero and the resulting output token amount will be zero. Secondly, the function is monotonically increasing, ensuring that a larger input token amount consistently results in a larger output token amount. Third, the property of concavity reflects the diminishing returns in token swaps. This implies that as the input amount $x$ increases, the additional output $f(x)$ gained from incremental unit of input decreases. In practical terms, this property implies that larger swaps typically receive a less favorable exchange rate than smaller ones. }

The swap function is characterized by three fundamental properties: 1) zero origin ($f(0)=0$), 2) monotonicity (strictly increasing),and 3) strict concavity. While the first two are obvious, strict concavity captures the critical economic principle of diminishing marginal returns. Mathematically, as the input amount $x$ increases, the marginal output price —represented by the derivative $f'(x)$, continuously declines. In financial terms, this implies that the marginal exchange rate worsens as the trade size grows. Consequently, the effective exchange rate (the average rate for the entire swap, $f(x)/x$) becomes less favorable for larger swaps compared to smaller ones. This phenomenon is termed price slippage~\cite{angeris2021analysis}. As an example, consider a constant product market maker (CPMM) pool (e.g., Uniswap V2) with reserves $R_{in}$ (input token) and $R_{out}$ (output token). The swap function $f(x)$, which maps an input $x$ to an output is $f(x) = \frac{R_{out} \cdot x}{R_{in} + x}$. The second-order derivative is 
$$
f''(x) = \frac{-2 R_{in}R_{out}}{(R_{in}+x)^3} < 0,
$$
which confirms that $f(x)$ is strictly concave. 
This property also holds for the end-to-end swap functions of concentrated liquidity AMMs like Uniswap V3~\cite{adams2021uniswap}. Moreover, in realistic scenario, the input amount $x$ is restricted to a finite range. On this bounded domain, the function's strict concavity implies that it is also strongly concave~\cite{boyd2004convex}.

Based on the swap function between token inputs and outputs, we can now introduce the token swap graph, a structure that provides the framework for modeling sequential token swaps.

\begin{definition} [Swap Graph]
\label{def: token swap graph}
A token swap graph is a directed graph \( G = (V, E, F) \), where: 
\begin{itemize}
    \item $V$ is a set of vertices, each of which represents a type of token. 
    \item $E \subseteq V \times V $ is a set of directed edges. For an edge, $e = 
    (v_i,v_j)\in E$, the edge represents a swap hop from token $v_i$ to token $v_j$.
    \item $F = \{ f_e | e \in E \}$ is the set of swap functions associated with the edges, where for each edge $e = (v_i, v_j)$, $f_e$  is the swap function. 
\end{itemize}
\end{definition}


In our model, the decentralized exchange ecosystem is represented as a directed graph. Each liquidity pool corresponds to two or several directed edges in this graph. Specifically, for a common case of a liquidity pool involving two tokens, $v_i$ and $v_j$, this pool generates two directed edges: one from $v_i$ to $v_j$ and vice versa. \eat{This bidirectionality is crucial since the swap function (and the corresponding exchange rate) for swapping from $v_i$ to $v_j$ may differ from that for swapping from $v_j$ to $v_i$ due to varying liquidity.} In particular, some liquidity pools may involve multiple tokens. To integrate into our model, these pools will introduce  distinct directed edges for every possible swap between any two tokens within the pool. As an example, a Curve 3pool~\cite{curve_3pool} contains  $3$ USD-related tokens and allows swaps between any two of them. This results in $3\times2=6$ directed edges in our model. Having formalized the representation of liquidity pools and the resulting token swap graph, we are now ready to define another crucial concept in our model: the Swap Path. 

\begin{table}[!t]
\centering
\caption{Frequently Used Notations}
\label{tab: notations}
\begin{tabularx}{0.48\textwidth}{lX}  
\toprule
\textbf{Notation} & \textbf{Description} \\
\midrule
$f,f_e$ & Swap function for DEX liquidity pool. \\
\( G=(V, E, F) \) & Token swap graph: \( V \) (set of tokens), \( E \) (set of directed swap hops), \( F \) (set of edge swap functions). \\
\( e = (v_i, v_j) \in E \) & A directed edge in \( G \), representing a swap hop from token \( v_i \) to token \( v_j \), associated with swap function \( f_e \). \\
\( s, t \in V \) & Source token \( s \) and target token \( t \) in the token swap graph. \\
\( x \) & Input amount of token \( s \) to swap. \\
\( \mathcal{P} \) & Set of chosen paths from \( s \) to \( t \) in \( G \). \\
\( W_p \) & Weight assigned to path \( p \in \mathcal{P} \), representing the fraction of input \( x \) allocated to \( p \). \\
\( f_P(x) \) & Output amount from path \( p \in \mathcal{P} \) for input \( x \). \\
\( w_e \) & Edge weight for edge \( e \in E_i \) in a hop of a path, with \( \sum_{e \in E_i} w_e = 1 \), determining the distribution of input across parallel edges. \\
\bottomrule
\end{tabularx}
\end{table}

\eat{
In this graph, each edge is associated to a liquidity pool. For most pools involving two tokens $(v_i,v_j)$ (e.g. Uniswap V2 and Uniswap V3), it generate two edges: a directed edge from $v_i$ to $v_j$ and a directed edge from $v_j$ to $v_i$. Since  the swap function from different direction in the same pool may have a different swap function due to the different liquidity. \hym{meaning that when user swap token i to token j, or verse versa.}For some liquidity pool involving multiple tokens (e.g., a Curve pool), the type of pool corresponds to multiple directed edges in the token swap graph. \hym{wait, any pool is one direction? seems all pools are two side} Specifically, for a pool containing $n$ types of tokens, say $ v_1, v_2, \ldots, v_n $, \hym{pool? whats the definition of pool, i mean one pool should only consider a pair of vi,vj right? why multiple types of tokens that have n tokens. }this pool will generate a directed edge $ (v_i, v_j) $ for every ordered pair of tokens $ (v_i, v_j) $ where $ i \neq j $. With the swap function and token swap graph now formalized, we proceed to define the core concept, \textbf{path}, in our model. It describes how tokens are swapped from a source token to a target token. We distinguish between two types of paths in the token swap graph: simple paths and complex paths.}

\begin{definition}[Swap Path]
\label{def: path}
A path in $ G $ from a source token $ v_0 $ to a target token $v_k$ is an ordered sequence of hops. A hop represents the token swap between two tokens, \rFour{consists} of one or more parallel edges in the graph. The number of hops $k$ is the length of path. We define two types of paths:
\begin{itemize}
    \item \rTwo{Single-Edge Path}: A single-edge path $ P(v_0, v_k) = \{e_1, e_2, \dots, e_k\} $ is an ordered sequence of edges, where each hop consists of a single edge $ e_i = (v_{i-1}, v_i) $.
    \item \rTwo{Multi-Edge Path}: A multi-edge path $ P_c(v_0, v_k) = \{E_1, E_2, \dots, E_k\} $ is an ordered sequence of parallel edges, where parallel edges $ E_i = \{e_{i,1}, e_{i,2}, \dots, e_{i,m_i}\} $ consists of one or more edges from $ v_{i-1} $ to $ v_i $.
\end{itemize}
\end{definition}
The swap path details the exchange of tokens from source to target. In our model, we define single-edge and multi-edge paths based on if there \rFour{exists} multiple edges in a hop of the path. As multi-edge paths offer richer opportunities for optimization, they are the primary focus of this paper. Therefore, in this paper, 'path' will refer to a multi-edge path unless otherwise noted. Given such a path $ P = \{E_1, E_2, \dots, E_k\} $ from a source token $ v_0 $ to a target token $ v_k $, the swap function of the path, denoted by $f_{P}(x) $, is the composition of the swap functions of each individual hop. Path swap function can be calculated as follows:

\begin{equation}
f_{P}(x) = f_{E_k}(f_{E_{k-1}}(\cdots f_{E_1}(x) \cdots)),
\end{equation}

and for each hop $E_i$, the swap function $f_{E_i}(x)$ depends on how we distribute the input token amount across the parallel edges within $E_i$. Therefore, we introduce edge weights to formalize the distribution. For each edge $e \in E_i$, we assign an edge weight $w_{e} \geq 0$ such that $\sum_{e\in {E_i}} w_{e} = 1$.  Then the swap function of the hop $E_i$ is given by:

\begin{equation}
\label{equ: edge weights}
f_{E_i}(x) = \sum_{e\in E_i}f_{e}(w_{e}\cdot x) \text{, where}\sum_{e\in {E_i}} w_{e} = 1,
\end{equation}

where $ x $ represents the amount of token entering hop $E_i$, and $f_{e}(x)$ is the swap function of edge $e$. Specifically, if the $i$-th hop has only one edge, its swap function is that of the edge itself. Using the definitions of swap functions and swap path, the function $f_P$ can describe the final token output. It's important to note that swap functions are strictly concave and monotonically increasing. This concavity is maintained throughout path composition and weighted hop sums, as both operations preserve concavity~\cite{boyd2004convex}. With these foundational concepts, we are now ready to formally define the token swap routing problem.

\eat{In the Token Swap Routing Problem, we aim to optimally route a given input quantity in paths from a source token to a target token through.  For each path, we also consider how to distribute the token amount across parallel edges within each hop of the path, if such parallel edges exist. To formalize this, we introduce two types of weights:
\begin{itemize}
    \item \textbf{Path Weight} (\(W_p\)): For each path \(p \in \mathcal{P}\) from the source token \(s\) to the target token \(t\), we assign a path weight \(W_p \geq 0\). This weight represents the fraction of the total input quantity \(x\) that is allocated to path \(p\). The sum of path weights across all chosen paths must be equal to 1, ensuring that the entire input quantity is utilized.
    \item \textbf{Edge Weight} (\(w_{p,e}\)): For each path \(p \in \mathcal{P}\) and each edge \(e\) within a hop of path \(p\), we assign an edge weight \(w_{p,e} \geq 0\). If a hop \(E_i\) in path \(p\) consists of multiple parallel edges, these weights determine how the token amount entering this hop is distributed among these parallel edges. For each hop \(E_i\) in path \(p\), the sum of edge weights for all edges \(e \in E_i\) must be equal to 1, ensuring that the entire token amount entering the hop is distributed.
\end{itemize}

With introduced weights, we now formally define the Token Swap Routing Problem as an optimization task over a directed graph. }

\begin{definition}[Token Graph Routing Problem] 
\label{def:tsrp}
Given a directed graph \( G = (V, E, F) \), a source token \( s \in V \), a target token \( t \in V \), an input quantity \( x > 0 \). We need to find a set of paths $\mathcal{P}$ from source \(s\) to target \(t\) in the graph. Each path \( p \in \mathcal{P} \) is assigned a weight \( W_p \geq 0 \), representing the fraction of the total input \( x \) allocated to that path, with the objective of maximizing the output. Within each path, the input at every hop is distributed across parallel edges according to edge weights \( w_e \geq 0 \), as specified in Equation \ref{equ: edge weights}. The optimization problem is formulated as:  
\begin{equation}  
\begin{aligned}  
\max_{\mathcal{P},W_p, w_e} &J(W)= \quad \sum_{p \in \mathcal{P}} f_p(W_p \cdot x), \\  
\text{s.t.} & \quad \sum_{p \in \mathcal{P}} W_p = 1, W_p \geq 0, \\  
& \quad \sum_{e \in E_i} w_e = 1, w_e \geq 0, \\
&\quad \forall E_i \in p, \forall p \in \mathcal{P}. \\    
\end{aligned}  
\end{equation}  

\end{definition}

Furthermore, an important constraint in this problem is that the chosen paths must be pool-disjoint; that is, no two paths in the set $\mathcal{P}$ can share edges that are associated with the same underlying pool. To explain, if two paths were to share the same pool, a token swap on one path would alter the state (e.g., the pool balances and thereby its swap function), which would affect the outcome of any subsequent swap involving that pool. Therefore, the pool-disjoint constraint ensures that each path's swap function remains independent. Notably, this assumption is consistent with the model proposed by~\cite{danos2021global}.






\eat{
    \item \textbf{Path Selection:} Find a set of paths \( P = \{P_1, P_2, \dots, P_k\} \) (where \( k \leq n \)) from the source token \( s \) to the target token \( t \) within the graph \(G\), such that each path \( P_i \) has a hop count of at most \( m \).
    \item \textbf{Optimal Input Allocation } Allocating the input \( x \) across the selected paths \( P_i \) in order to maximize the total output. Specifically, the equation is presented as follows:
}



\eat{An implicit but important constraint in this problem is that the edges involved in these paths must be pool-disjoint. As introduced in Definition \ref{def: token swap graph}, a pool may correspond to multiple different edges. If different paths share edges that are associated with the same pool, executing a swap along one path would alter the pool’s state, thereby changing the output of the swap function for any subsequent path using that pool. The pool-disjoint constraint ensures that each path’s swap function remains independent. The assumption is aligned with the model proposed in~\cite{danos2021global}. This provides a lower bound for the Token Swap Routing Problem.}





\rThree{

\section{Existing Solutions Revisited}
\label{sec:existing}
\eat{
In this section, we will review the existing solutions for solving the Token Graph Routing Problem(TGRP). Before discussing the limitations of specific algorithms, we review the main challenges in solving the problem. First, the scale of the graph is massive. As detailed in Section \ref{sec:intro}, the graph contains at least 400,000 tokens and 380,000 liquidity pools. Second, the problem suffers from extreme numerical heterogeneity. The exchange rate between assets can differ by orders of magnitude between different tokens (e.g., BTC at $10^{5}$ USD vs. SHIB at $10^{-6}$ USD).}

As highlighted in Section \ref{sec:intro}, solving the Token Graph Routing Problem (TGRP) requires overcoming four-fold challenges: non-linear edge weights which invalidate optimal substructure, massive graph scale, extreme numerical heterogeneity, and latency constraints. \rFour{In this section, we review the landscape of existing solutions against these conflicting constraints. While various approaches exist—ranging from classical algorithms to general-purpose convex solvers—they all succumb to specific bottlenecks, failing to simultaneously address the full spectrum of challenges. Consequently, none serve as a viable solution for production environments. We categorize these methods into three primary families and analyze why they fall short of solving the TGRP.}

\spara{Classical Graph Algorithms} 
Fundamentally, the TGRP constitutes a specific instance of the Generalized Network Flow (GNF) problem. The theoretical standard for solving such problems lies in classical Primal-Dual methods.(e.g., the Successive Shortest Path algorithm~\cite{truemper1973optimal}) These algorithms operate on a rigorous economic logic: they maintain node potentials (representing the current market valuation of each token) and iteratively search for paths where the actual exchange rate offered by the network exceeds this valuation. While this provides a sound strategy for discovering profitable routes, these algorithms are fatal in the context of TGRP. A classical GNF algorithm, blind to this convexity, would erroneously push flow into a path until hard capacity limits are reached, failing to stop when the marginal price becomes unfavorable. Other classical graph algorithms, such as shortest-path algorithms (e.g., Yen's algorithm~\cite{yen1970algorithm}), are also ill-suited for TGRP as they fail to capture the non-linear nature of the edge weights.

\spara{Global Convex Optimization}
Previous studies proposed to formulate the TGRP as a global convex optimization problem~\cite{angeris2022optimal,angeris2022constant,diamandis2023efficient,diamandis2024convex}. While mathematically rigorous, these numerical methods encounter two fatal bottlenecks when applied to real-world DEXs. First, they struggle with model expressivity. Methods based on general solvers~\cite{angeris2022optimal,angeris2022constant} are constrained by canonical forms and cannot efficiently encode the complex formulas like concentrated liquidity mechanisms~\cite{adams2021uniswap}. Subsequent approaches~\cite{diamandis2023efficient, diamandis2024convex} using dual decomposition result in non-smooth objective functions, leading to convergence difficulties on complex topologies~\cite{github_issue_29}. \rTwo{Second, and more critically, they fail to ensure strict flow conservation. Numerical methods (e.g., interior-point methods) operate within an error tolerance, accepting solutions that are close enough rather than exact. In markets with extreme asset heterogeneity, this lack of strict equality leaves residual tokens (`dust') stranded at intermediate hops. This phenomenon is validated experimentally in Section \ref{subsec:convex}. Cleaning up these residuals would require handling routing problems for each residual token. In summary, while convex optimization offers a theoretical solution, it is too slow to satisfy sub-second latency constraints and too numerically unstable to ensure atomic on-chain execution, thus failing to address two of the four critical challenges of TGRP.}

\spara{Industrial Heuristics} 
To address these challenges in practice, the industry standard, such as the Uniswap Smart Order Router (SOR)~\cite{uniswap_routing}, 
employs a heuristic-based approach. The algorithm first uses heuristic rules to select liquidity pools with high total value locked (TVL). It then performs path searching on this subgraph and use combinatorial optimization to find the optimal allocation. However, this method has two major drawbacks. First, the heuristic filtering is aggressive and often excludes pools with lower liquidity but higher price. Second, the combinatorial optimization phase relies on splitting the input amount into discrete parts to test combinations. This brute-force nature makes the algorithm computationally expensive, yet it still fails to guarantee a mathematically optimal solution.

}

\begin{figure*}[htbp]
    \centering
    \includegraphics[width=0.9\linewidth]{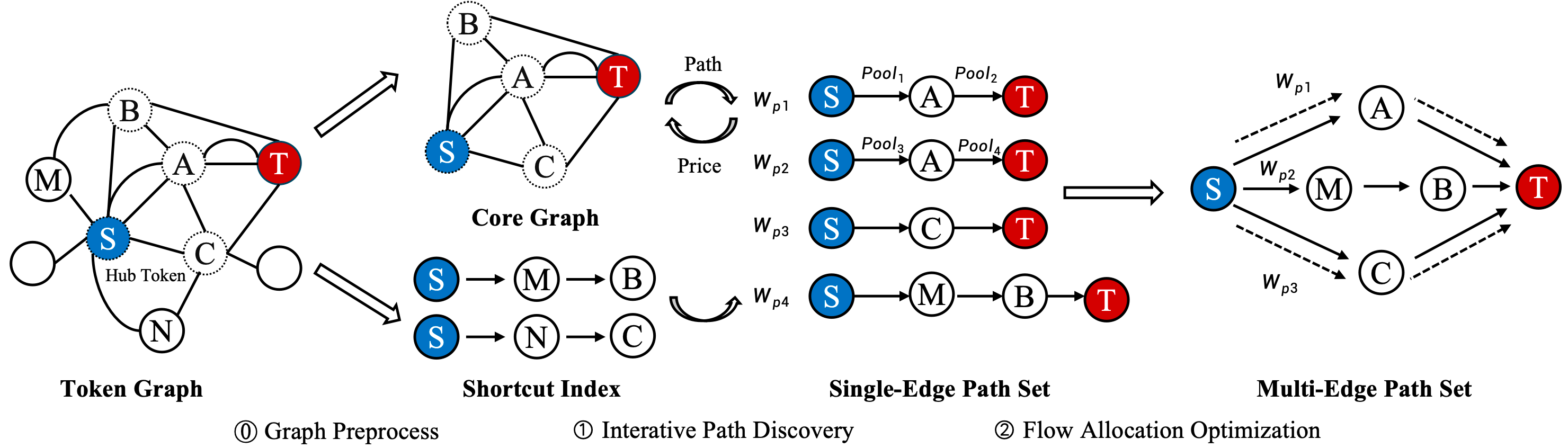}
    \caption{\rTwo{The Workflow of PRIME Algorithm}}
    \label{fig:prime_algo}
\end{figure*}

\eat{
\subsection{Limitations of Prior Algorithms} 
Existing academic approaches to the TGRP can be broadly classified into two paradigms: theoretical optimization models and practical heuristic methods. The first category, represented by convex optimization-based methods~\cite{angeris2022constant,angeris2022optimal, diamandis2023efficient,diamandis2024convex}, offers a \rFour{theoretical} framework. However, their practical application is hampered by significant numerical instability. As demonstrated in our analysis in Section \ref{subsec:convex}, this instability arises from a fundamental precision mismatch between the floating-point arithmetic of mathematical solvers and the fixed-point arithmetic of on-chain smart contracts. This discrepancy often renders their theoretically optimal solutions practically infeasible, necessitating complex and often suboptimal post-processing to be executable. The second paradigm includes practical, heuristic-based \rFour{algorithms} like Uniswap's Smart Order Router (SOR)\cite{uniswap_routing}. These \rFour{algorithms} simplify the complex continuous allocation problem by discretizing it—for instance, by splitting trades into fixed increments. While computationally efficient, this heuristic fundamentally constrains the solution space. By design, it prevents the discovery of a true global optimum, leading to the performance and efficiency gaps we identified in our evaluation. Finally, while commercial aggregation services provide sophisticated routing in practice~\cite{0x_swap,1inch_swap}, their algorithms are proprietary and closed-source. This 'black box' nature prevents a direct and reproducible academic comparison, and they were therefore excluded as baselines in our study.

Diamandis et al.~\cite{diamandis2023efficient} proposed an efficient algorithm for the optimal routing problem across CFMMs. Their approach formulates the problem as maximizing a general user utility function, $U(\Psi)$, where $\Psi \in \mathbb{R}^n$ is the trade vector representing the net change in the user's holdings across all $n$ assets in the system after executing trades across all available CFMMs. Their algorithm leverages dual decomposition to decouple the problem into parallelizable subproblems, each solving for optimal arbitrage trades within individual CFMMs.

The algorithm introduces dual variables (interpreted as asset prices), $\nu \in \mathbb{R}^n$, to relax the coupling constraint $\Psi = \sum_{i=1}^m A_i \Delta_i$, where $\Delta_i$ is the trade vector in market $i$'s local assets and $A_i$ maps it to the global asset space. This relaxation decomposes the original complex problem into simpler, independent subproblems:

\begin{enumerate}
\item An optimization probelm over the network trade vector: $\max_{\Psi} (U(\Psi) - \nu^T \Psi)$. This finds the best overall network outcome $\Psi$ given the prices $\nu$.
\item An independent optimal arbitrage problem for each market $i$: $\max_{\Delta_i \in T_i} (A_i^T \nu)^T \Delta_i$. This finds the most profitable trade $\Delta_i$ within market $i$, given the local prices $A_i^T \nu$ derived from the global prices $\nu$. This can be solved independently for each CFMMs.
\end{enumerate}

The overall algorithm then focuses on finding the \textit{clearing} prices $\nu^*$ (by solving the dual problem) such that the solutions $\Psi^*$ and $\{\Delta_i^*\}$ obtained from these subproblems also satisfy the original coupling constraint $\Psi^* = \sum A_i \Delta_i^*$, thereby solving the original routing problem.

While theoretically elegant, this formulation contains critical oversights when applied to practical token swaps. For example, in a actual token swap, the net change for all intermediate tokens must be exactly zero. That is, if token C is used in the path A $\rightarrow$ C $\rightarrow$ B, the final network trade vector $\Psi$ must satisfy $\Psi_C = 0$. However, the dual decomposition method does not inherently enforce these zero-balance constraints on intermediate tokens. The optimal $\Psi^*$ resulting from solving the dual problem (and the corresponding $\Delta_i^*$ from the subproblems) is optimal with respect to the overall utility function $U(\Psi)$ and the relaxed coupling constraint, but it does not guarantee that specific $\Psi_j^*$ for intermediate $j$ will be zero. While one could attempt to encode these zero constraints into the utility function $U$ (e.g., by setting $U(\Psi) = -\infty$ if $\Psi_j \neq 0$ for any intermediate $j$), doing so effectively coupling between the markets. Consequently, the primary computational advantage of the decomposition method is lost when addressing the hard constraints required by practical token swaps. Therefore, their algorithm can not be applied to real-world token swap problem.
}

\section{The PRIME Algorithm}
\label{sec:algorithm}

Following the solutions reviewed in Section~\ref{sec:existing}, we present \textbf{PRIME}, our algorithm for solving the Token Graph Routing Problem.

\subsection{Overview of Algorithm and Rationale}
\label{subsec: Overview of Algorithm}
\rTwo{
In this section, we firstly provide an overview of our algorithm. PRIME is a two-stage framework supported by a graph preprocessing phase to address the scalability of the problem. The overal workflow is introduced as follows:

\spara{Stage 0: Graph Preprocessing} In the off-line stage 0, we preprocess the token graph $G=(V,E,F)$ into two components: 
\begin{enumerate}
    \item \textbf{Core Graph}: We define a set of Hubs $\mathcal{H} \subset V$, consisting of top-$K$ assets connected to pools with highest total liquidity (e.g., WETH, USDC, USDT). For any routing query from source token $s$ to target token $t$, we execute the procedure \textsc{InduceCoreGraph} to construct the subgraph $G_{core}$ induced by the vertices $\mathcal{H} \cup \{s, t\}$. This core graph can capture the main routing paths from the source to the target.
    \item \textbf{Shortcuts Index}: While $G_{core}$ captures the main routing paths, the path with higher price through non-core tokens not included in $\mathcal{H}$. We define a \textit{Shortcut} as a single-edge path $P = (h_i, v_1, \dots, v_n, h_j)$ where $h_i, h_j \in \mathcal{H}$ and intermediate nodes $v \notin \mathcal{H}$. We employ the \textsc{BuildShortcutIndex} function to pre-compute these connections and store them in the index $\mathcal{I}_{sc}$. During routing, $\mathcal{I}_{sc}$ allows the algorithm to integrate these superior shortcuts in constant time.
\end{enumerate}

\begin{algorithm}[t]
    \DontPrintSemicolon
    \SetKwInOut{Input}{Input}
    \SetKwInOut{Output}{Output}
    \Input{Graph $G=(V,E,F)$, Source $s$, Target $t$, Input Amount $x$, Hubs $\mathcal{H}$}
    \Output{Optimal Allocation $(\mathcal{P}_{multi}, W^*, \{w_e\}^*)$}
    \BlankLine
    \tcp{Stage 0: Graph Preprocess }
    $\mathcal{I}_{sc} \gets \textsc{BuildShortcutIndex}(G, \mathcal{H})$\;
    $G_{core} \gets \textsc{InduceCoreGraph}(\mathcal{H}, s, t)$\;
    
    \BlankLine
    \tcp{Stage 1: Iterative Path Discovery}
    Initialize: $\mathcal{P}_{single} \gets \emptyset$, Threshold $\tau \gets 0$\;
    \While{true}{
        $p_{new} \gets \textsc{FindPath}(G_{core} \cup \mathcal{I}_{sc}, s, t, \tau)$\;
        \If{$p_{new} = \text{null}$ \textbf{or} $\text{MarginalPrice}(p_{new}) < \tau$}{
            \textbf{break}\;
        }
        $\mathcal{P}_{single} \gets \mathcal{P}_{single} \cup \{p_{new}\}$\;
        $(\Delta_{out}, \tau) \gets \textsc{ASGM}(\mathcal{P}_{single}, x)$\;
    }
    
    \BlankLine
    \tcp{Stage 2: Optimization and Solution}
    $\mathcal{P}_{multi} \gets \textsc{MergeAndExpand}(\mathcal{P}_{single}, \mathcal{I}_{sc})$\;
    $W^*, \{w_e\}^* \gets \textsc{ASGM}(\mathcal{P}_{multi}, x)$\;

    \Return{$(\mathcal{P}_{multi}, W^*, \{w_e\}^*)$}
    \caption{\textsc{PRIME}}
    \label{alg:prime_overview}
\end{algorithm}

\spara{Stage 1: Iterative Path Discovery} Based on the graph preprocessing results, the first stage of the routing algorithm aims to identify a candidate paths set $\mathcal{P}_{single}$. We perform a Breadth-First Search (BFS) pruning algorithm (Algorithm \ref{alg:bfs-pruning}) on the core graph $G_{core}$. The algorithm iteratively adds a new path to the set only if it provides a marginal price improvement over the current threshold $\tau$. This threshold is updated using the Adaptive Sign Gradient Method (ASGM) in each iteration. Finally, the algorithm checks the shortcuts index $\mathcal{I}_{sc}$ if any shortcut can provide a better price between hubs in the candidate paths and update the path set.

\spara{Stage 2: Optimization and Solution} In the second stage, the algorithm proceeds to optimize the marginal price at a finer granularity (Lines 10-11). It first performs a topological transformation via \textsc{MergeAndExpand}, which merges paths sharing identical token sequences and expands every hop with alternative parallel edges. This transforms the initial single-edge candidates into a multi-edge path set $\mathcal{P}_{multi}$. Finally, we employ the Adaptive Sign Gradient Method (ASGM) to strictly optimize the allocation of the input amount across these paths and parallel edges (Line 11), maximizing the total output until the marginal prices converge.

The computational complexity of PRIME is optimized by restricting path discovery to the Core Graph. The total time complexity is $O(K \cdot |V_{core}| \cdot |E_{core}| + K^3 \log(1/\epsilon))$, where $K$ is the number of identified paths, $|V_{core}|$ and $|E_{core}|$ are the number of vertices and edges in the core graph, respectively. A key design advantage is the early termination criterion in Stage 1, which halts the process once new paths no longer yield a marginal price improvement. This ensures efficiency for small trades where the optimal number of paths $K$ is small.

}

\begin{algorithm}[!bpt]
    \DontPrintSemicolon
    \SetKwInOut{Input}{Input}
    \SetKwInOut{Output}{Output}
    
    \Input{Graph $G_{core}$, Source $s$, Target $t$, Amount $x$, Threshold $\tau$}
    \Output{A Single-Edge Path $p^*$ or \text{null}}
    \BlankLine
    $Q.\text{push}(s, x, \emptyset)$ and $Best[\cdot] \gets 0$\;
    \While{$Q$ is not empty}{
        $(u, \Delta_{cur}, p) \gets Q.\text{pop()}$\;
        \If{$u == t$}{
            \If{$(\Delta_{cur} / x) > \tau$}{ \Return $p$ }
            \textbf{continue}\;
        }
        \ForEach{edge $e=(u, v) \in G_{core}$}{
            $\Delta_{next} \gets \rFive{f_e}(e, \Delta_{cur})$\;
            \If{$\text{len}(p) < m$ \textbf{and} $\Delta_{next} > Best[v]$}{
                $Best[v] \gets \Delta_{next}$\;
                $Q.\text{push}(v, \Delta_{next}, p + e)$\;
            }
        }
    }
    \Return \text{null}
    \caption{\rFive{\textsc{FindPath}}}
    \label{alg:bfs-pruning}
    \end{algorithm}

\subsection{Path Discovery}
\label{subsec: Path Discovery}
To efficiently discover a path with a better marginal price from the source to the target, we introduce an efficient \textsc{FindPath} algorithm (Algorithm \ref{alg:bfs-pruning}) via a Breadth-First-Search (BFS) pruning approach. While a naive enumeration of all paths is computationally infeasible due to the exponential growth of cycles in the token graph, our approach adapts the Shortest Path Faster Algorithm (SPFA)~\cite{ahuja1990faster} for the token swap problem. SPFA optimizes the standard Bellman-Ford algorithm by utilizing a queue for breadth-first traversal and updating (``relaxing'') a node only when a strictly better path is discovered. In our implementation, we maintain a map $Best[\cdot]$ to record the maximum output achieved for each token during the search. Any path reaching a token with an output lower than the current best is pruned. For instance, if there are two different paths from the source to an intermediate token $A$, where the first yields 10 units and the second yields only 8, the algorithm immediately discards the second path and only extends the first. This dominance pruning significantly reduces the search space by selectively exploring only promising branches.

The time complexity of this algorithm is bounded by $O(|V_{core}|\cdot|E_{core}|)$. \rFive{However, in practical DeFi scenarios, we enforce a maximum hop count $m$ (typically 3 or 4) based on empirical observations: within the densely connected Core Graph, paths exceeding this depth suffer from excessive slippage. This constraint, combined with dominance pruning, rapidly reduces the candidate set within the first few layers.} Consequently, the average complexity approaches $O(k|E_{core}|)$ where $k$ is a small constant relative to $m$, allowing the algorithm to efficiently identify the optimal path.

\begin{algorithm}[!bpt]
    \DontPrintSemicolon
    \SetKwProg{Fn}{Function}{}{}
    \SetKwInOut{Input}{Input}
    \SetKwInOut{Output}{Output}
    \Input{Discovered Paths $\mathcal{P}^*$ \\
        Input quantity \( x > 0 \),\\ Minimum Learning Rate \(\alpha\) \\
        Max Iterations \( T \), Decay Factor \(\beta\)
    }
    \Output{Optimal Weights $W^*$}
    \BlankLine
    Initialize: $W_0 \gets \left[\frac{1}{|\mathcal{P}|}, \dots, \frac{1}{|\mathcal{P}|}\right]$, $t \gets 0$\;
    $\mathrm{SIGN} \gets \mathbf{0}$\;
    \While{$t < T$ \textnormal{ and $W_{t}$ not converged}}{
      \ForEach{path $p \in \mathcal{P^*}$}{
        Optimize \rFive{$w_e$ for edges in path $p$} via ASGM\;
        Compute marginal price: $g_p \gets f_p'(W_p \cdot x)$
      }
      Identify paths: 
        $p^- \gets \arg\min g_p$\ and $p^+ \gets \arg\max g_p$\;
      
      $\mathrm{SIGN}_{p^+,p^-} \gets [1,-1]$\;
      $\Delta \gets \Delta_0$
      
      \Repeat{$J(W_{t} + \Delta \cdot\mathrm{SIGN} ) \geq J(W_{t}) + \alpha \Delta\cdot x(g_{p^+} - g_{p^-})$}
      {
        $\Delta = \Delta \times \beta$
      }
      $t \gets t + 1$\;
      $W_{t} \gets W_{t-1} + \mathrm{SIGN}\cdot\Delta$\;
      $\mathrm{SIGN} \gets \mathbf{0}$\;
    }
        $W^{*}\gets W_t$
    \caption{\textsc{ASGM}}
    \label{alg:asg}
  \end{algorithm}

\subsection{Allocation Optimization}
\label{subsec: Allocation Optimization}

After introducing the above method to efficiently discover the path, the set of paths $\mathcal{P}$ is fixed. We now focus on the optimization problem of allocating the swap amount across each path and at each hop. According to Definition~\ref{def:tsrp}, the Token Graph Routing Problem can be reduced to the following Token Allocation problem:
\begin{equation}  
\begin{aligned}  
\max_{W_p, w_e} J(W) &=  \sum_{p \in \mathcal{P}} f_p(W_p \cdot x), \\  
\text{s.t.}  \quad  \sum_{p \in \mathcal{P}} W_p &= 1, W_p \geq 0, \\  
\quad \sum_{e \in E_i} w_e &= 1, w_e \geq 0. \\
\end{aligned}  
\end{equation}

Here, $W = [W_p]_{p \in \mathcal{P}}$ is the vector of allocation weights for each path, and $f_p(\cdot)$ is the output function for path $p$. As introduced in \ref{def: swap function}, the swap function are strongly concave. Since the sum of (strongly) concave functions preserves (strong) concavity, the overall objective function $J(W)$ is strongly concave. The constraint $\Omega = \{\sum_{p \in \mathcal{P}} W_p = 1, W_p \geq 0; \rFive{\sum_{e \in E_i} w_e = 1, w_e \geq 0, \forall i}\}$ forms a high-dimensional simplex, \rFive{where $E_i$ denotes the set of parallel edges at hop $i$ of a path}. The optimality condition for this problem can be characterized as follows:

\begin{lemma}[Optimal Condition]
\label{lem:optimal_condition}
A feasible allocation for the token allocation problem is optimal if and only if the marginal prices are equal across all components (i.e., paths and parallel edges).\end{lemma}
\begin{proof}
The proof of Lemma~\ref{lem:optimal_condition} is straightforward. If marginal prices are unequal across paths or parallel edges with positive allocation, reallocating input from components with lower marginal prices to those with higher marginal prices would increase the objective value, contradicting optimality. Furthermore, since the objective function is strictly concave and the feasible region is a convex set, according to Karush–Kuhn–Tucker's condition~\cite{boyd2004convex}, the optimal solution is unique. This guarantees that the allocation achieving equal marginal prices is the only global optimum.
\end{proof}

The primary challenge in solving this optimization problem is the vast scale difference in token values. Specifically, the price of Bitcoin (BTC) can be on the order of $10^{5}$ USD, while that of Shiba Inu (SHIB) is around $10^{-6}$ USD. \rTwo{From the standpoint of numerical optimization, this $10^{11}$ scale disparity renders the objective function ill-conditioned~\cite{boyd2004convex}. Specifically, the Hessian matrix of the objective function exhibits a massive condition number, causing the curvature to vary drastically across different dimensions. In such ill-conditioned landscapes, traditional gradient-based methods fail.} To address this, we propose the \textbf{Adaptive Sign Gradient Method (ASGM)}, a robust algorithm that is robust to such scale differences. The core idea of ASGM is to continuously shift swap amounts from paths with worst price to those with better price. At every step, the algorithm identifies two paths: the least efficient (highest cost) and the most efficient (lowest cost) path. By rebalancing an optimal amount of tokens from the less efficient path to the more efficient one, ASGM gradually reduces price differences across all paths. This process repeats until all paths become equally efficient—meaning no further output can be made by reallocating tokens. At this point, the system reaches an optimal overall swap price. 

Algorithm~\ref{alg:asg} introduces ASGM in detail. First, the total swap amount is equally distributed across all paths and parallel edges. During each iteration, the algorithm computes the marginal price $g_p$ for every path $p$ and identifies two critical paths: $p^+$ (the path with the highest marginal price) and $p^-$ (the path with the lowest marginal price). The goal is to iteratively adjust allocations to minimize the price gap in the paths. A direction vector $\text{SIGN}$ specifies the adjustment: $p^+$ reduces its allocation while $p^-$ increases its allocation. The adjustment amount is controlled by an adaptive step size $\Delta$, which is initially set to a large pre-determined value for exploration and decays exponentially until satisfying the Armijo condition~\cite{armijo1966minimization}:
\begin{equation}
    \begin{aligned}
        J(W_{t} + \Delta \cdot \text{SIGN}) &\geq J(W_{t}) + \alpha \Delta \cdot \nabla J(W_t)^\top \text{SIGN}, \\
        \nabla J(W)^\top \cdot \text{SIGN} &= x(g_{p^+} - g_{p^-}),
    \end{aligned}
\end{equation}
where $J$ is the objective function and $W_t$ is the allocation at iteration $t$. The parameter $\alpha \in (0,1)$ dictates the minimum fraction of progress required. The Armijo condition is a criterion ensuring that the step size $\Delta$ achieves a sufficient increase in the objective function. It verifies that the actual improvement gained is at least a fraction $\alpha$ of the improvement predicted by the function's linear approximation. Our algorithm employs a backtracking line search, starting with a large $\Delta$ and decaying it until this condition is met, thus preventing overly large steps that might overshoot the optimum while still ensuring meaningful progress in each iteration. This learning-rate-free design is the core innovation of ASGM, making it inherently robust to the vast scale differences in token values. In addition to this, we also analyze the convergence properties of ASGM and establish its convergence rate.

\begin{theorem}[Convergence of ASGM]
\label{thm:convergence}
For the Token Allocation Problem with a strongly concave objective function, the Adaptive Sign Gradient Method (ASGM) converges to the optimal allocation $W^*$ at a linear rate.
\end{theorem}

\begin{proof}
We prove the theorem relying on the strong concavity of the objective function and the simplex constraint. For problems with this specific structure, it can be shown that the improvement gained in each step is proportional to the current optimality gap. This proportionality forces the gap to decrease exponentially, which represents a linear convergence rate. The theoretical foundation for the linear convergence of this class of update algorithms was established by Lacoste-Julien and Jaggi~\cite{lacoste2015global}. A complete step-by-step derivation and detailed convergence analysis for ASGM are provided in Appendix~\ref{app:convergence}.
\end{proof}

\section{Experiment}
\label{sec:experiment}


In this section, we comprehensively evaluate the performance of our proposed algorithm based on a series of experiments on real-world blockchain data. Our evaluation is designed to assess three key dimensions: \textbf{Effectiveness}, measured by the execution price; \textbf{Efficiency}, measured by the algorithm's computation time; and \eat{\textbf{Robustness}} \textbf{Stability}, demonstrated by consistent performance across diverse and volatile market conditions. We benchmark PRIME against several baseline algorithms to systematically validate its advantages. Additional sensitivity analyses of ASGM hyperparameters are deferred to Appendix~\ref{app:ablation_lr_decay}.

\subsection{Experiment Setting}
\label{subsec: experiment setting}

\rTwo{
\spara{Dataset and Market Condition} Our experiments are grounded in historical data retrieved from an Ethereum mainnet archival node. To ensure our evaluation covers diverse market dynamics, we collected complete DEX pool state snapshots from three distinct one-week periods: a bearish market (Mar '25), a stable market (May '25), and a bullish market (Jul '25). We focus on the two dominant CFMM algorithms: the Constant Product formula (V2)~\cite{uniswapv2core} and the Concentrated Liquidity model (V3)~\cite{adams2021uniswap}, covering major protocols including Uniswap, SushiSwap, and PancakeSwap~\cite{sushiswap,pancakeswap}. To ensure calculation accuracy, we implement the swap logic in Rust using the \textit{alloy} crate's \texttt{U256} type, which provides sufficient precision to mirror the fixed-point arithmetic of Solidity.

Table~\ref{tab:dataset_stats} summarizes the dataset statistics. While the raw on-chain data is massive, directly routing on it is impractical. In practice, we observe that 96.65\% of tokens connect exclusively to WETH. These tokens act as leaf nodes and are pruned to reduce the search space. Furthermore, we exclude \textit{fee-on-transfer} tokens~\cite{nolo2024feontransfer}, as these tokens will charge tax on each transfer, causing discrepancies between the CFMM-calculated output and the actual execution amount. Consequently, the effective routing graph is smaller, consisting of approximately 11,000 tokens and 25,000 pools.
}

\begin{table}[!tbp]
    \centering
    \caption{\rTwo{Dataset Statistics (in thousands). The table details the raw graph size and pool distribution across protocols (V2/V3) under different market conditions.}}
    \label{tab:dataset_stats}
    \renewcommand{\arraystretch}{1.3} 
    \setlength{\tabcolsep}{4pt}
    \begin{tabular}{l c c cc cc cc}
    \toprule
    \multirow{2}{*}{\textbf{Market}} & \textbf{Nodes} & \textbf{Edges} & \multicolumn{2}{c}{\textbf{Uniswap}} & \multicolumn{2}{c}{\textbf{Sushi}} & \multicolumn{2}{c}{\textbf{Pancake}} \\
    \cmidrule(lr){4-5} \cmidrule(lr){6-7} \cmidrule(lr){8-9}
        & (Tokens) & (Pools) & V2 & V3 & V2 & V3 & V2 & V3 \\
    \midrule
    Bearish & 402.7k & 428.4k & 392.2k & 30.5k & 4.1k & 0.4k & 0.7k & 0.4k \\
    Stable  & 418.7k & 445.1k & 406.8k & 31.8k & 4.2k & 0.5k & 1.4k & 0.5k \\
    Bullish & 448.5k & 475.6k & 425.6k & 40.3k & 4.2k & 0.5k & 4.3k & 0.6k \\
    \bottomrule
    \end{tabular}
    \end{table}

\spara{Environment}
Experiments are performed on a server running Ubuntu 22.04.5. The server hardware includes a 10-core Intel Xeon Silver 4210 CPU operating at 2.20GHz. It is equipped with 64 GiB of DDR4 RAM running at 2666 MHz. Our algorithm utilizes the Rust 1.92.0 runtime environment for execution and we use Gurobi 12.0.3 to solve the convex optimization problem. 

\spara{Pairs Selection}
To assess the algorithms’ generalizability, we selected a set of four tokens representing distinct categories of crypto assets. We include WBTC (Wrapped Bitcoin) and USDT (Tether) as foundational assets representing high market cap and a dominant stablecoin, respectively. We also choose ChainLink (LINK), one of the largest utility tokens. To capture assets heavily influenced by market sentiment, we include PEPE, which is the top meme token by market capitalization on the Ethereum blockchain. Unlike the other tokens, which feature deep and fragmented liquidity across numerous pools, PEPE's liquidity is primarily concentrated in a few large, dominant pools. For each token in this set, we test trades against WETH across a logarithmic scale of input amounts: $1$, $10$, $100$, and $1,000$ WETH. This range allows us to analyze algorithmic performance in retail to institutional scales.

\subsection{Baselines}
\label{subsec:Baselines}
We conduct a rigorous performance evaluation of PRIME, benchmarking it against established algorithms and internal variants to validate our design choices. Our comparison covers four key categories: a state-of-the-art academic single-path algorithm,the industry standard, a general flow-based variant, an ablation version of our algorithm, and a theoretical convex optimization method. The performance gains are quantified in basis points (bp). A basis point is a standard financial unit representing one-hundredth of a percentage point (0.01\%).

\rThree{
\spara{Optimal Single Path: Line Graph Routing} To evaluate the necessity of our multi-path routing framework, we compare our approach against the optimal single-path routing algorithm proposed by Zhang et al.~\cite{zhang2025line}. By transforming the token graph into a line graph representation, this method guarantees the identification of the single path with the highest global execution price. This baseline provides the theoretical upper bound of routing output without splitting input across multiple paths.
}

\begin{table*}[!htbp]
    \centering
    \sisetup{group-separator={,}}
    \small 
    \setlength{\tabcolsep}{0pt}
    
    \caption{
        \rFour{
            Performance comparison of routing algorithms across different trade sizes for WETH. 
            \textbf{OSP} serves as the baseline. 
            For SOR, PRIME-Flow, PRIME (No Shortcut), and PRIME, prices are shown with the basis point (bp) difference relative to OSP in parentheses. 
            Execution time is reported in milliseconds (ms).
            For readability, prices for WBTC are presented as value \(\times 10^3\), and for PEPE as value in millions.
        }
    }
    \label{tab:main_results_final}
    
    \newcommand{\valbp}[2]{#1 \scriptsize{(#2)}}
    
    \begin{tabular*}{\textwidth}{@{\extracolsep{\fill}} l lllrlrlrlr }
    \toprule
    
    & & \rThree{\textbf{OSP~\cite{zhang2025line}}} & \multicolumn{2}{c}{\textbf{SOR~\cite{uniswap_routing}}} & \multicolumn{2}{c}{\rFour{\textbf{PRIME-Flow}}} & \multicolumn{2}{c}{\textbf{PRIME-Core}} & \multicolumn{2}{c}{\textbf{PRIME}} \\
    \cmidrule{3-3} \cmidrule{4-5} \cmidrule{6-7} \cmidrule{8-9} \cmidrule{10-11}
    
    \multirow{2}{*}{\textbf{Token}} & \multirow{2}{*}{\textbf{WETH}} & \textbf{Price} & \textbf{Price} & \textbf{Time} & \textbf{Price} & \textbf{Time} & \textbf{Price} & \textbf{Time} & \textbf{Price} & \textbf{Time} \\
    & & \scriptsize{(Baseline)} & \scriptsize{(bp diff.)} & \scriptsize{(ms)} & \scriptsize{(bp diff.)} & \scriptsize{(ms)} & \scriptsize{(bp diff.)} & \scriptsize{(ms)} & \scriptsize{(bp diff.)} & \scriptsize{(ms)} \\
    \midrule
    
    \textbf{USDT} & 1    & 3,018.37 & \valbp{3,019.65}{+4.24} & 283.9 & \valbp{3,023.11}{+15.71} & 920.5 & \valbp{3,021.32}{+9.78} & 24.0 & \valbp{3,022.33}{+13.11} & 37.1 \\
    & 10   & 3,016.81 & \valbp{3,017.35}{+1.78} & 279.4 & \valbp{3,017.53}{+2.36}  & 365.5 & \valbp{3,017.35}{+1.78} & 20.0 & \valbp{3,019.35}{+8.39}  & 21.9 \\
    & 100  & 3,014.05 & \valbp{3,015.15}{+3.63} & 284.7 & \valbp{3,015.22}{+3.86}  & 348.6 & \valbp{3,015.18}{+3.73} & 22.4 & \valbp{3,015.18}{+3.74}  & 27.7 \\
    & 1000 & 2,980.19 & \valbp{3,001.99}{+73.15}& 283.1 & \valbp{3,004.13}{+80.32} & 873.5 & \valbp{3,003.07}{+76.79}& 34.2 & \valbp{3,003.41}{+77.93} & 39.2 \\
    \midrule

    \textbf{WBTC} & 1    & 25.24 & \valbp{25.24}{+2.35}  & 258.7 & \valbp{25.26}{+11.79} & 345.4 & \valbp{25.26}{+7.90} & 27.3 & \valbp{25.27}{+12.77} & 33.9 \\
    & 10   & 25.22 & \valbp{25.23}{+2.16}  & 233.5 & \valbp{25.23}{+2.58}  & 274.6 & \valbp{25.23}{+2.07} & 12.8 & \valbp{25.23}{+2.30}  & 22.2 \\
    & 100  & 25.16 & \valbp{25.19}{+12.05} & 253.3 & \valbp{25.19}{+12.92} & 225.1 & \valbp{25.19}{+12.58}& 26.8 & \valbp{25.19}{+12.58} & 32.2 \\
    & 1000 & 24.74 & \valbp{25.03}{+116.04}& 333.0 & \valbp{25.04}{+123.45}& 430.9 & \valbp{25.04}{+120.99}& 44.9 & \valbp{25.04}{+121.81}& 53.4 \\
    \midrule

    \textbf{LINK} & 1    & 187.88 & \valbp{187.88}{0.00} & 260.4 & \valbp{188.17}{+15.14} & 743.3 & \valbp{188.10}{+11.37} & 26.6 & \valbp{188.17}{+14.97} & 27.7 \\
    & 10   & 187.84 & \valbp{187.85}{+0.31} & 382.7 & \valbp{187.86}{+0.85}  & 398.0 & \valbp{187.86}{+0.68}  & 11.1 & \valbp{187.86}{+0.78}  & 11.5 \\
    & 100  & 187.44 & \valbp{187.44}{0.00}  & 370.8 & \valbp{187.45}{+0.60}  & 257.9 & \valbp{187.44}{+0.45}  & 5.7  & \valbp{187.45}{+0.62}  & 8.3 \\
    & 1000 & 183.01 & \valbp{183.01}{0.00}  & 372.7  & \valbp{183.11}{+5.80}  & 249.0 & \valbp{183.11}{+5.73}  & 11.7 & \valbp{183.15}{+7.96}  & 12.4 \\
    \midrule

    \textbf{PEPE} & 1    & 244.97 & \valbp{244.98}{+0.47} & 107.3 & \valbp{245.09}{+4.97}  & 133.7  & \valbp{245.09}{+4.87}  & 7.2 & \valbp{245.00}{+4.96}  & 10.3 \\
    & 10   & 244.61 & \valbp{244.72}{+4.31} & 111.2 & \valbp{244.72}{+4.37}  & 78.9 & \valbp{244.72}{+4.32}  & 2.6  & \valbp{244.72}{+4.32}  & 6.1 \\
    & 100  & 242.17 & \valbp{242.81}{+26.34}& 110.9 & \valbp{242.81}{+26.39} & 92.0  & \valbp{242.81}{+26.39} & 2.2  & \valbp{242.81}{+26.39} & 6.2 \\
    & 1000 & 220.20 & \valbp{225.73}{+250.90}& 116.1 & \valbp{222.38}{+259.69} & 206.6 & \valbp{225.88}{+258.07}& 12.3 & \valbp{225.92}{+259.53}& 12.5 \\
    \bottomrule
    \end{tabular*}
\end{table*}

\rFour{
\spara{Path Constraint Analysis: PRIME-Flow} To evaluate the efficiency of our parallel-path design, we introduce PRIME-Flow. This is a relaxed variant of PRIME that permits overlapping paths during the discovery phase. The algorithm iteratively identifies an augmenting path with a better marginal price and merges it into the existing flow. We employ a ternary search algorithm to determine the optimal split ratio of the input amount between the current flow and the new path~\cite{press2007numerical}. The process terminates when no path with a better marginal price can be found. This baseline serves to quantify the yield sacrificed by our pool-disjoint constraint to achieve higher computational efficiency.

}

\spara{Shortcut Index Ablation: PRIME-Core} To validate the effectiveness of the Shortcut Index, we evaluate PRIME-Core. This variant executes the standard PRIME workflow but restricts path discovery exclusively to the Core Graph. By comparing this variant against the full PRIME algorithm, we demonstrate that our Shortcut design can capture better paths through non-core tokens with negligible computational overhead.

\spara{Industry Standard: Uniswap SOR} To evaluate the practical performance of our approach, we compare it against the industry standard, Uniswap Smart Order Router (SOR)~\cite{uniswap_routing}. SOR employs a two-phase approach. First, it uses heuristics to select pools with deep liquidity and then performs a Breadth-First Search (BFS) to identify a limited set of candidate paths. Second, in the allocation phase, it transforms the routing challenge into a combinatorial optimization problem. To achieve this, SOR pre-computes the output for each path given various discretized portions of the total input (e.g., $5\%$, $10\%$, etc.). With these pre-computed values, SOR then solve a combinatorial optimization problem to find the optimal combination of splits that maximizes the total output. This baseline represents the current state-of-the-art in industrial applications.

\spara{Theoretical Optimum: Convex Optimization}
To establish a theoretical upper bound on performance, we implement the convex optimization approach from Angeris et al.~\cite{angeris2022optimal}. This method models the entire swap as a single, large-scale optimization problem. \eat{The challenge is that the exact trading function of each CFMM pool (e.g., the constant product $R_x \cdot R_y = K$ for Uniswap V2) is a non-linear equality constraint, which renders the overall problem non-convex.} The key insight is to relax these non-convex equalities into convex inequalities. For a Uniswap V2 pool, this relaxation takes the form of a geometric mean constraint: the equality $R_x \cdot R_y = K$ is replaced with the convex inequality $\sqrt{R_{x'}R_{y'}} \geq \sqrt{R_xR_y}$, where $R_{x'}$ and $R_{y'}$ are the new token reserves. This relaxation transforms the entire routing problem into a Second-Order Cone Program (SOCP)~\cite{boyd2004convex}, which is a standard form of convex optimization. Consequently, we can solve this SOCP for the globally optimal solution using modern conic solvers like Gurobi~\cite{Gurobi}.

\subsection{Main Result}
\label{subsec:Main Result}

This section evaluates the performance of PRIME against established baselines, using the results summarized in Table \ref{tab:main_results_final}. \rFour{To ensure a fair and reproducible basis for our analysis, we execute the algorithms using a one-hour time-averaged window starting at $00:00$ UTC+08:00 on July 14th, 2025.} The analysis is organized to highlight the advantages of our design choices compared to industrial standards and internal variants.

\begin{figure*}[!t]
    \centering
    
    \begin{subfigure}[b]{0.48\textwidth}
        \centering
        \includegraphics[width=\textwidth]{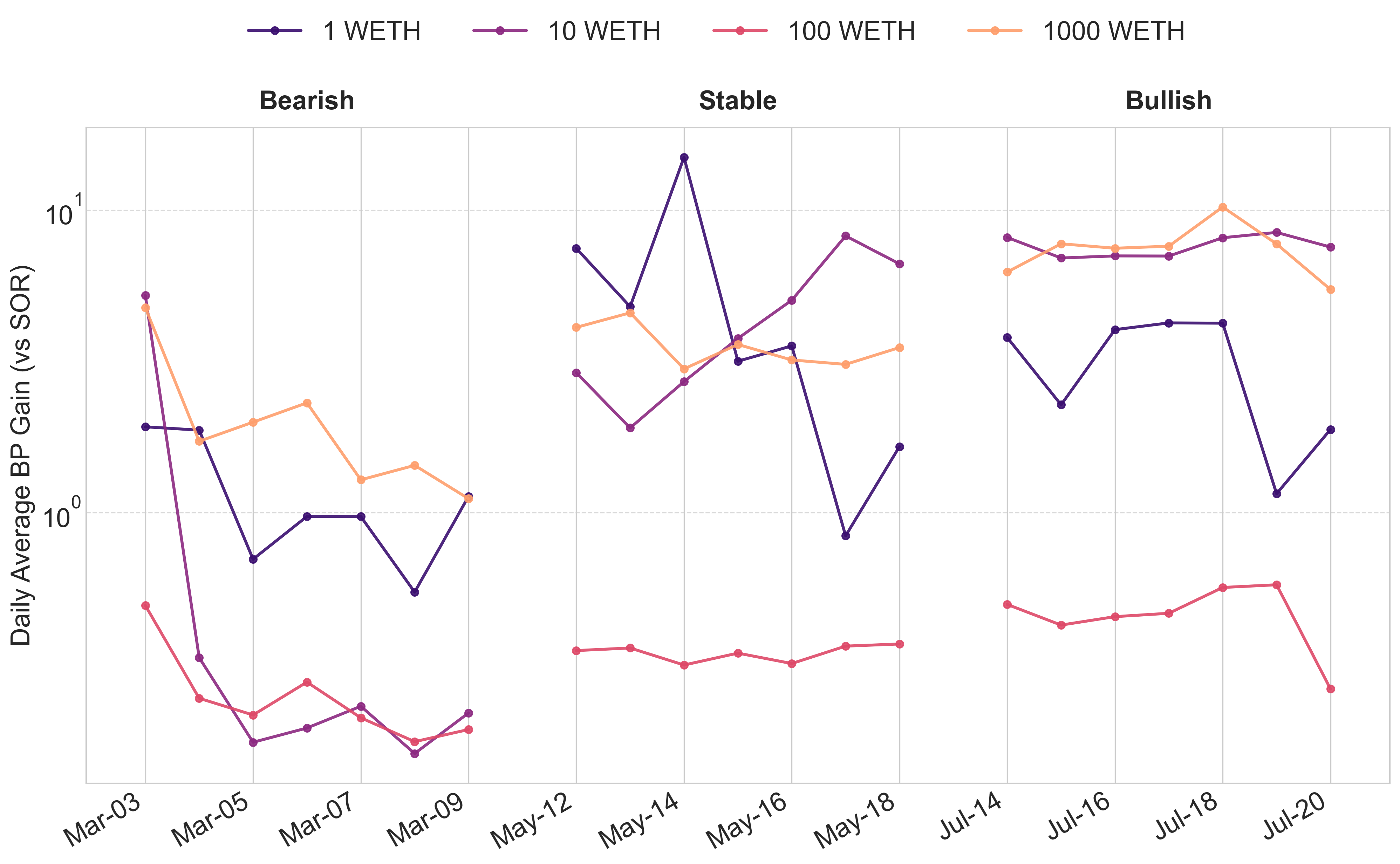}
        \caption{Price Improvement Performance} 
        \label{subfig:bp_gain_robustness}
    \end{subfigure}
    \begin{subfigure}[b]{0.48\textwidth}
        \centering
        \includegraphics[width=\textwidth]{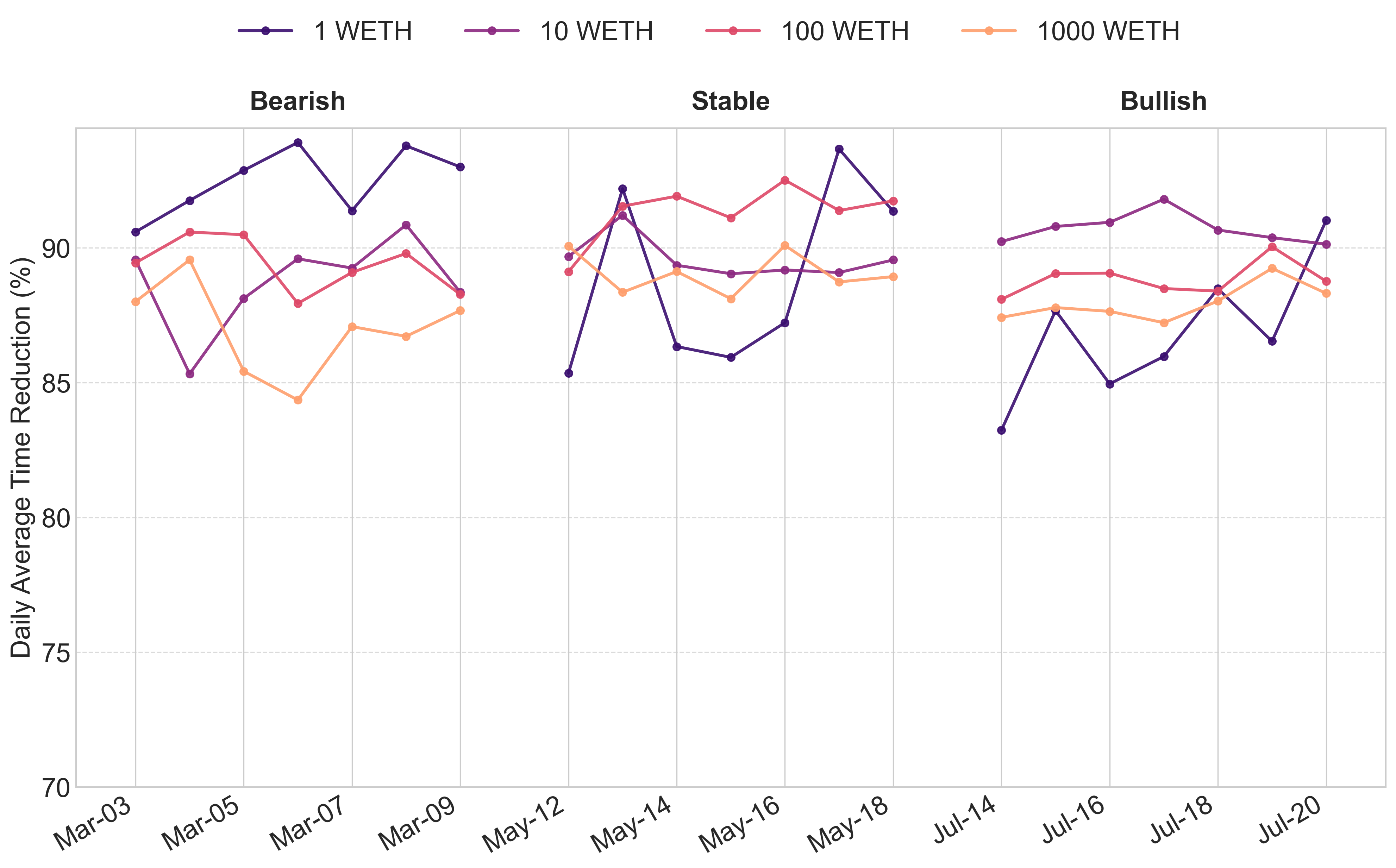}
        \caption{\rThree{Time Reduction Performance}} 
        \label{subfig:time_reduction_robustness}
    \end{subfigure}
    
    \caption{Performance of PRIME relative to SOR across three distinct market periods: Bearish, Stable, and Bullish. The figure plots the daily average performance for WETH to USDT swaps across a range of input amounts. (a) Price improvement of PRIME over SOR, measured in basis points (bp). (b) Percentage reduction in computation time achieved by PRIME.}
    \label{fig:robustness_analysis}
\end{figure*}

\spara{Comparison with Mainstream Algorithms} 
\rThree{
Table \ref{tab:main_results_final} reveals a fundamental trend: all multi-path routing algorithms (SOR and PRIME variants) consistently outperform the single-path OSP baseline across varying trade scales. This confirms the necessity of splitting flow across diverse pools to minimize slippage.} Among these multi-path solutions, PRIME demonstrates a decisive advantage over the industry-standard SOR in both execution price and computational efficiency. Regarding price, a critical observation is PRIME's price advantage in small-scale trades of 1 WETH. For example, in USDT swaps, PRIME enables relative price improvements of up to 14.97 bp in LINK and 10.42 bp in WBTC swaps over SOR's baseline. \eat{These results validate that PRIME's Shortcut Index identifies high price long-tail paths that SOR's heuristic rules may overlook.} As trade volumes scale to 1,000 WETH, PRIME's optimization continues to widen the absolute gap, outperforming SOR by an additional 8.63 bp for volatile assets like PEPE. Computationally, PRIME continusely exhibits efficiency regardless of trade size. It reduces execution time by 86.9\% for 1 WETH USDT swaps and by 96.7\% for large-scale LINK trades relative to SOR. Therefore, we conclude that compared to the current Industry Standard, PRIME is capable of finding superior routing solutions with shorter execution time.

\rFour{
\spara{Ablation Studies} To validate the architectural choices of PRIME, we compare it againts two internal variants: PRIME-Flow and PRIME-Core. First, we compare PRIME with the relaxed PRIME-Flow variant. While the relaxed flow-based variant occasionally achieves higher execution prices by allowing overlapping paths (e.g., +80.32 bp vs. +77.93 bp for 1,000 WETH USDT). In several instances, PRIME significantly outperforms the flow-based method, notably in the 10 WETH USDT scenario where PRIME achieves +8.39 bp while PRIME (Flow) only reaches +2.36 bp. This is because flow-based methods often fail to converge stably to the global optimum compared to our parallel-path optimization. PRIME sacrifices a negligible average profit of less than 1 bp in cases where Flow performs better, but in exchange, it reduces computational overhead by an average of 93.1\%, with several cases exceeding a 95\% reduction in execution time(e.g., 39.2 ms vs. 873.5 ms for 1,000 WETH USDT).

Second, we compare PRIME with the PRIME-Core variant. The full PRIME algorithm consistently provides higher prices relative to the PRIME-Core variant, with the advantage being particularly pronounced in small-scale of 1 WETH trades. The shortcut Index allows PRIME to achieve +13.11 bp for USDT and +12.77 bp for WBTC, whereas the PRIME-Core variant only secures +9.78 bp and +7.90 bp, respectively. In summary, according to our experimental results, we sacrifice a marginal amount of execution time (less than 10 ms) to secure price improvements, validating the Shortcut Index as a cost-effective component for enhanced liquidity discovery.

}

\subsection{Stability Across Market Conditions}
To assess its stability across different market conditions, We sampled market states at 10-minute intervals over three periods introduced in Section \ref{subsec: experiment setting}. This process yielded a total of 3,024 market snapshots. We report results for the WETH-USDT pair, as it effectively represents performance trends observed across other major assets.

\spara{Price Comparison} 
Figure~\ref{subfig:bp_gain_robustness} demonstrates that PRIME consistently outperforms SOR, with the most significant advantages appearing at the extremes of trade size (1,000 and 1 WETH). For large 1,000 WETH trades, PRIME achieves gains around 10 bps, particularly in bullish conditions. This confirms the algorithm's ability to mitigate significant price slippage through optimal multi-path allocation. Conversely, for small 1 WETH swaps, PRIME also delivers outsized returns, occasionally spiking above 10 bps during stable periods. This performance is attributed to the \textbf{Shortcut Index}, which identifies paths through "long-tail" pools that industrial heuristics aggressively prune. Medium-sized trades (10 and 100 WETH) show more modest but positive gains.


\begin{table}[!tbp]
    \centering
    \scriptsize
    \caption{Comparison of PRIME and the adjusted convex solution across varying graph sizes. The table quantifies PRIME's performance advantage in basis points (bp) and presents the raw numerical imbalances from the convex solver that require post-processing adjustment. }
    \label{tab:convex_results}
    \begin{tabular}{@{}lcccc@{}}
    \toprule
    \multirow{2}{*}{\makecell{Graph Size \\ (\# Pools)}} & \multicolumn{2}{c}{Output (USDT)} & \multirow{2}{*}{\makecell{PRIME's \\ Advantage (bp)}} & \multirow{2}{*}{\makecell{Solver's Raw \\ Imbalance*}} \\
    \cmidrule(lr){2-3} 
    & PRIME & Adjusted Convex & & \\ 
    \midrule
    
    
    \makecell{3 Tokens \\ (6 Pools)} & 2.53253 M & 2.53257 M & -0.16 & \makecell[l]{USDC: +22.4579} \\ 
    \midrule
    
    \makecell{15 Tokens \\ (60 Pools)} & 2.53253 M & 2.53320 M & -2.64 & \makecell[l]{WBTC: -0.0064 \\ AAVE: +0.1202 \\ stETH: +0.0046 \\ and 7 other assets} \\ 
    \midrule
    
    \makecell{25 Tokens \\ (117 Pools)} & 2.53253 M & 2.53238 M & +0.59 & \makecell[l]{WBTC: -0.0627 \\ YFI: -0.0400 \\ stETH: -0.0111 \\ and 14 other assets} \\ 
    \bottomrule
    \end{tabular}
    \end{table}

\rThree{\spara{Time Comparison} Figure~\ref{subfig:time_reduction_robustness} highlights PRIME's computational efficiency, maintaining a substantial time reduction (generally 85--95\%) across all periods. Notably, the 1 WETH category exhibits the highest variability rather than consistent speed dominance. While it achieves peak reductions in bearish markets, the time reduction fluctuates significantly in active (Stable/Bullish) markets, occasionally dropping to around 83\%. This variance occurs because active markets generate a denser set of candidate shortcuts; evaluating these additional high-quality possibilities slightly increases computation time compared to stable market states.}

\subsection{Comparison with Theoretical Optimum}
\label{subsec:convex}

Beyond industry benchmarks, we evaluate PRIME against the theoretically optimal solution provided by convex optimization~\cite{diamandis2023efficient,angeris2022optimal}. While theoretically elegant, this approach is sensitive to the numerical precision required for real-world routing. We designed a targeted experiment to explicitly reveal the practical limitations of this algorithmic approach. This experiment is limited to Uniswap V2 pools, as the solver cannot handle Uniswap V3's concentrated liquidity mechanism. We constructed a series of routing problems of increasing scale, from a small 3-token graph (6 pools) to a large 25-token graph (117 pools), each simulating a 1,000 WETH for USDT swap.
\begin{figure}[!tbp]
    \centering
    \includegraphics[width=0.9\linewidth]{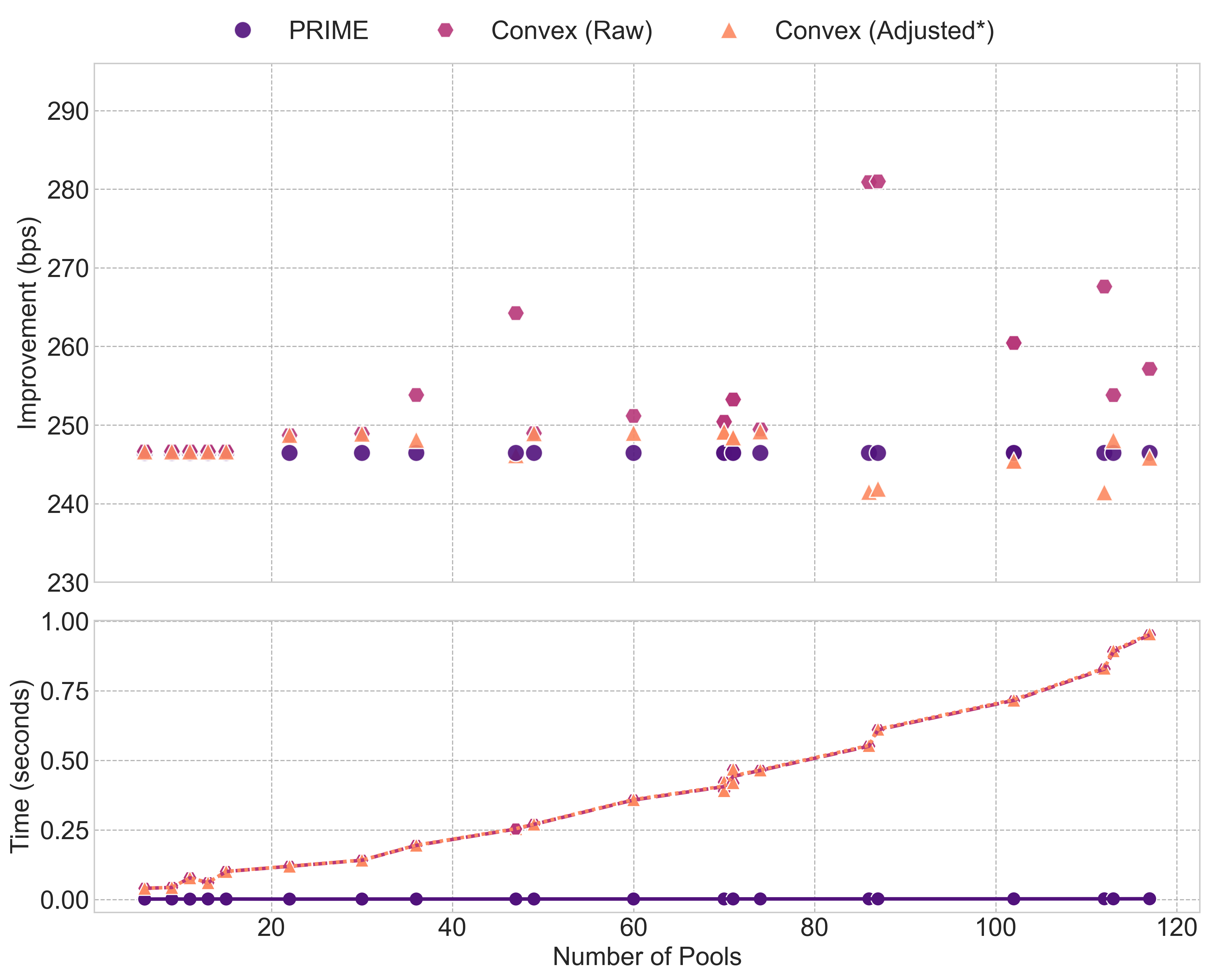}
    \caption{Scalability Comparison}
    \label{fig:convex comparison}
    	\vspace{-3mm}
\end{figure}

\spara{Price Comparison} As shown in Table~\ref{tab:convex_results}, raw convex outputs are practically inexecutable due to numerical faults—small residual balances. For instance, in the 25-token graph, the solution leaves deficits in WBTC and YFI while creating surpluses in stETH and 14 other assets. This instability stems from the discrepancy between the the Ethereum Virtual Machine (EVM~\cite{wood2014ethereum})'s 256-bit fixed-point arithmetic and the solver's 64-bit float representation, which loses up to six digits of precision during conversion. To enable a valid comparison, we use our path discovery algorithm to post-process the raw convex output into a feasible solution. The top panel of Figure \ref{fig:convex comparison} compares PRIME's output against both the raw and the adjusted outputs. While raw theoretical output appears slightly higher, applying a corrective heuristic to ensure feasibility significantly reduces its yield.Ultimately, the adjusted convex output offers a negligible advantage of at most 3 bps, and PRIME even secures superior results in more complex 117-pool scenarios.

\spara{Time Comparison}.
The time comparison in Figure \ref{fig:convex comparison} highlights a critical difference in scalability. PRIME exhibits near-constant time performance (Purple line), with its computation time remaining negligible and irrespective of the graph size. In contrast, a convex solver's computation time scales linearly with the number of pools (Orange line). As can be seen, when there exists $20$ pools, convex method spends nearly $0.1$ seconds for computation. However, when there are $100$ pools in the graph, the computation time climbs to nearly $1$ second. This scaling behavior makes the convex approach significantly less practical than PRIME's constant-time alternative—especially considering this is on a simplified dataset where the convex solver is even viable. In a full, real-world scenario, it would be computationally intractable.

\eat{

Beyond industry benchmarks, we evaluate PRIME against the theoretically optimal solution provided by convex optimization~\cite{diamandis2023efficient,angeris2022optimal}. While theoretically elegant, this approach is known to be sensitive to the scale and numerical precision required for real-world problems~\cite{github_issue_29}. We designed a targeted experiment to explicitly reveal the practical limitations of this algorithmic approach. This experiment is limited to Uniswap V2 pools, as the solver cannot handle Uniswap V3's concentrated liquidity mechanism. We constructed a series of routing problems of increasing scale, from a small 3-token graph (6 pools) to a large 25-token graph (117 pools), each simulating a 1,000 WETH for USDT swap.

\begin{figure}[!htbp]
    \centering
    \includegraphics[width=0.9\linewidth]{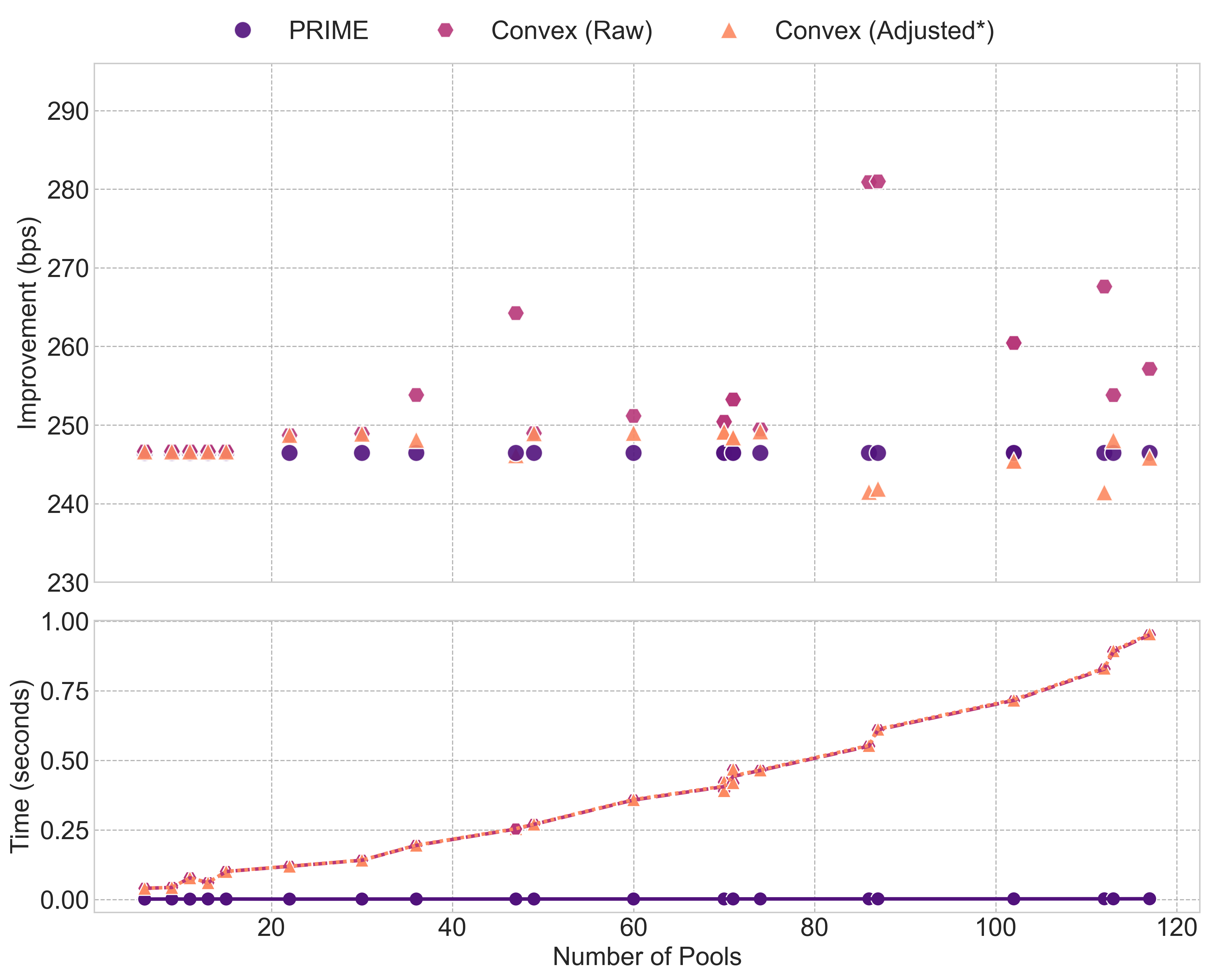}
    \caption{Scalability Comparison}
    \label{fig:convex comparison}
    	\vspace{-3mm}
\end{figure}

\spara{Instability of Solution}
As shown in Table~\ref{tab:convex_results}, the raw output from the convex solver is not directly executable due to numerical faults. These faults manifest as small residual balances—deficits or surpluses—in intermediate tokens. For instance, in the 25-token graph, the solution leaves deficits in WBTC and YFI while creating surpluses in stETH and 14 other assets. These faults stem from a fundamental discrepancy in numerical representation: the Ethereum Virtual Machine (EVM)~\cite{wood2014ethereum} uses 256-bit fixed-point arithmetic (uint256), whereas convex solvers operate with 64-bit float numbers, which can only represent 15-17 significant decimal digits. However, on-chain token quantities frequently exceed this limit. For instance, a WETH reserve with 18 decimals has 22 significant digits. When this value is loaded into a float64, it is truncated to conform to the 15-17 digit limit, losing the last six digits of precision. The necessary conversion introduces unavoidable precision losses, forcing the solver to optimize against slightly inaccurate constraints. To render the solution executable, we designed a post-processing heuristic that corrects these imbalances. It adjusts trade volumes to resolve deficits and routes surpluses back to the target token via our path discovery algorithm (Section~\ref{subsec: Path Discovery}).

\spara{Price Comparison}
The top panel of Figure \ref{fig:convex comparison} compares PRIME's output against both the raw and the adjusted outputs of the convex solver. While the solver's raw theoretical output initially appears superior to PRIME's, it is practically infeasible due to token imbalances caused by numerical precision issues. After applying our corrective heuristic to produce a feasible result, this initially high output drops significantly.
The comparison of feasible results is telling: the adjusted convex output offers a negligible advantage, outperforming PRIME by at most $3$ bps. In more complex scenarios, such as the $117$-pool case, PRIME's output is actually superior. This shows that once the solution feasibility is adjusted, the theoretical advantage of the convex approach becomes minimal.

\spara{Time Comparison}
The time comparison in Figure \ref{fig:convex comparison} highlights a critical difference in scalability. PRIME exhibits near-constant time performance (Purple line), with its computation time remaining negligible and irrespective of the graph size. In contrast, a convex solver's computation time scales linearly with the number of pools (Orange line). As can be seen, when there exists $20$ pools, convex method spends nearly $0.1$ seconds for computation. However, when there are $100$ pools in the graph, the computation time climbs to nearly $1$ second. This scaling behavior makes the convex approach significantly less practical than PRIME's constant-time alternative—especially considering this is on a simplified dataset where the convex solver is even viable. In a full, real-world scenario, it would be computationally intractable.

}
\begin{figure*}[htbp]
  \centering
  \begin{subfigure}[b]{0.3\linewidth}
    \centering
    \includegraphics[width=\textwidth]{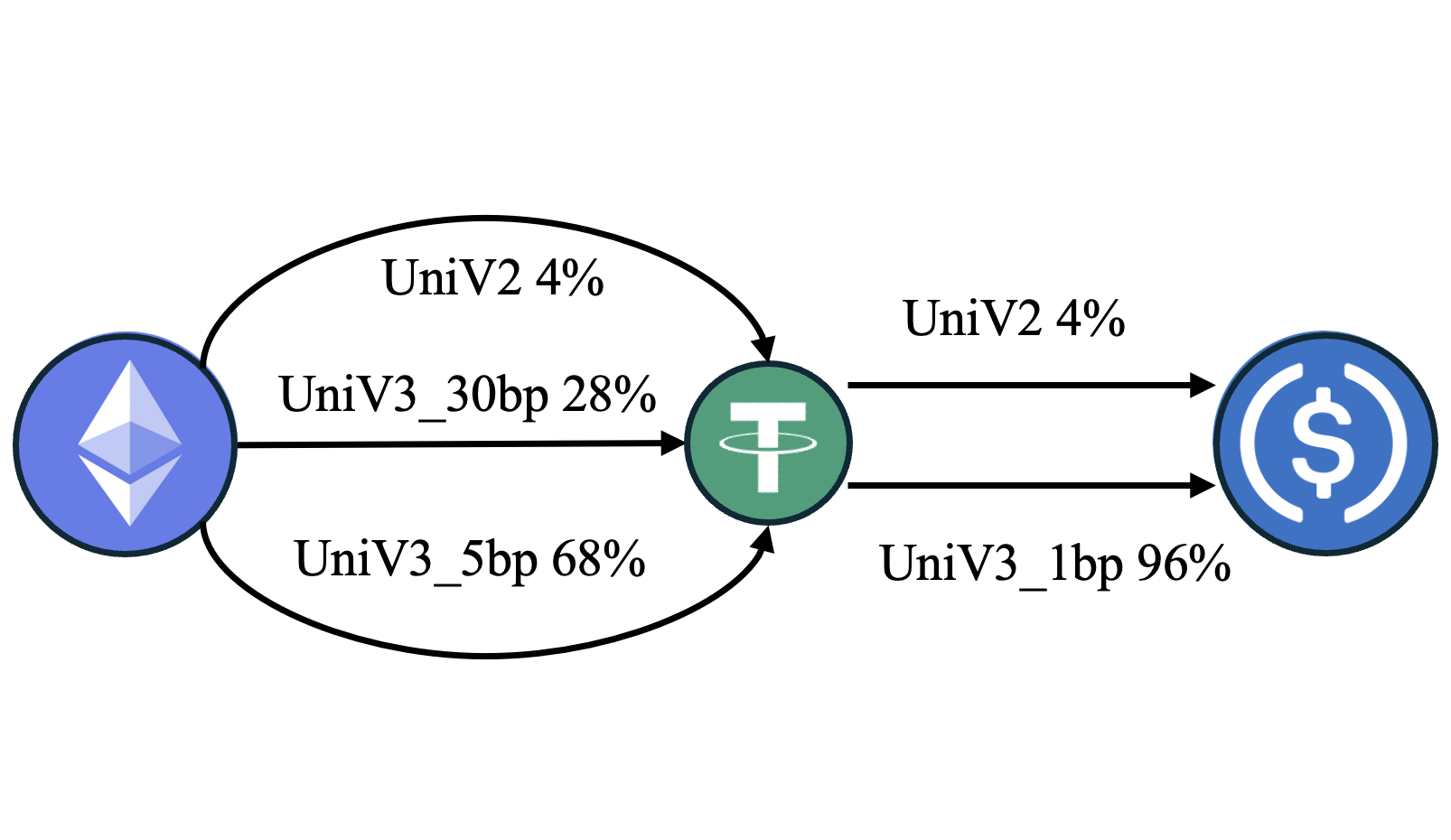}
    \caption{Single Multi-Edge Path (DODO \cite{dodo2020next})}
    \label{fig:SCP}
  \end{subfigure}
  \hfill
  \begin{subfigure}[b]{0.34\linewidth}
    \centering
    \includegraphics[width=\textwidth]{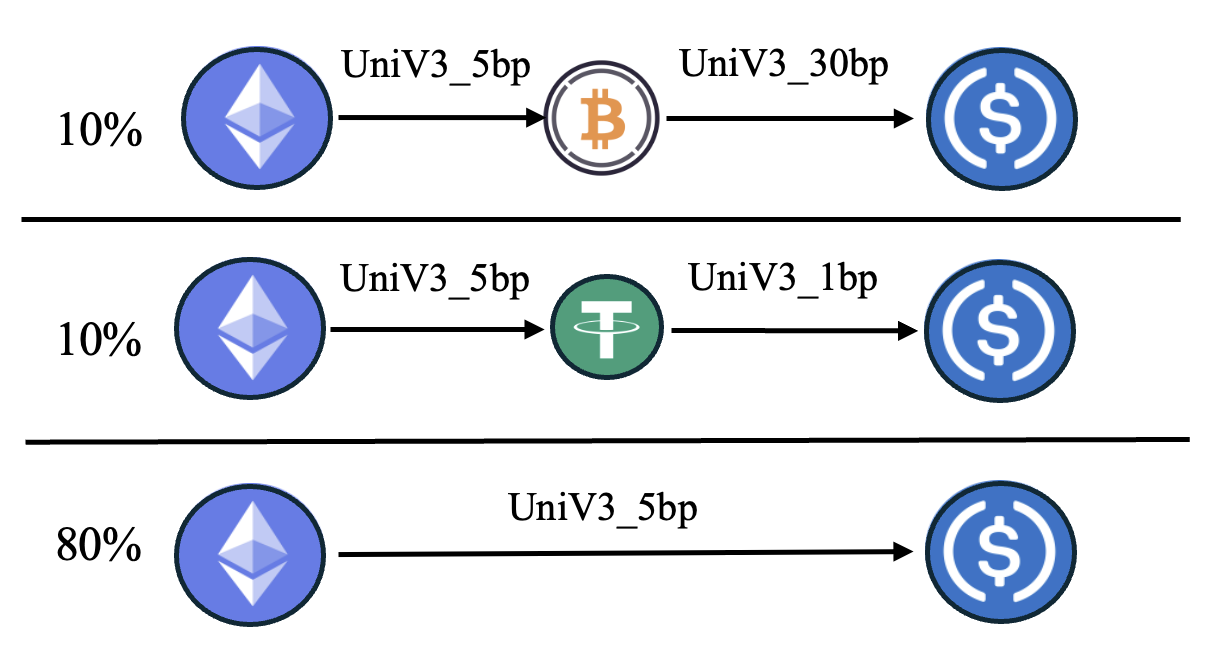}
    \caption{Multiple Single-Edge Paths (Uniswap \cite{uniswap_routing})}
    \label{fig:MSP}
  \end{subfigure}
  \hfill
  \begin{subfigure}[b]{0.34\linewidth}
    \centering
    \includegraphics[width=\textwidth]{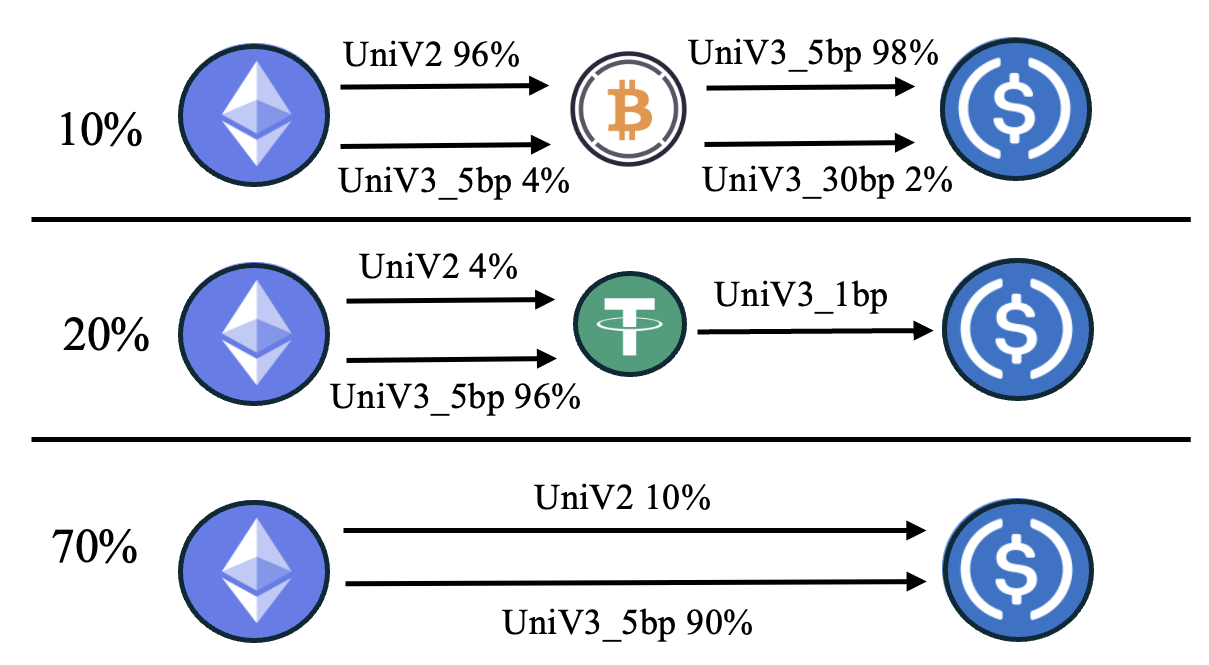}
    \caption{Multiple Multi-Edge Paths (1inch \cite{1inch_fusion})}
    \label{fig:MCP}
  \end{subfigure}

  \caption{
    Token swap execution patterns in production protocols. Uniswap V3 offers multiple pools for the same token pair, distinguished by fee tiers specified in the label (e.g., `UniV3\_5bp'). The standard 30bp fee for V2 pools is omitted for clarity. Percentages denote flow allocation.
  }

  \label{fig:swap-patterns}
\end{figure*}

\section{Discussion}
\label{sec:Discussion}

\spara{Practicality of Modeling} A key advantage of our model is its direct alignment with the structure of on-chain swap transactions, ensuring that our algorithmic solutions are practically executable. This alignment is rooted in \textit{Swap Path} introduced in Definition~\ref{def: path}, which is expressive enough to represent the diverse routing strategies employed by major DeFi protocols, as illustrated in Figure \ref{fig:swap-patterns}. For instance, protocols like DODO~\cite{dodo2020next} route the entire input into a single multi-edge path. As shown in Figure~\ref{fig:SCP} the hop between WETH and USDT utilized three parallel edges, while the subsequent hop between USDT and USDC utilized two parallel edges. In contrast, Uniswap split the swap amount across three distinct, disjoint paths as Figure~\ref{fig:MSP}. Each of these paths was structurally single-edge, meaning every hop within each path utilized only a single edge. Besides, other specialized aggregation protocol like 1inch~\cite{1inch_swap} distribute the input across multiple disjoint path, where any of these path leverages parallel edges for intermediate hops. 

In summary, these representative execution patterns are all encompassed by our generalized \textit{Swap Path} definition. Because our model accurately captures the problem structure, these protocols can directly integrate the PRIME algorithm to solve their routing challenges. This affirms the immediate practical value of our approach for building and optimizing token swap transactions in DeFi.

\spara{Extensibility to Heterogeneous Markets} In this paper, we focus on the token graph routing problem on Ethereum blockchain. However, our model is designed for broad applicability and can readily generalize to more markets, including multi-chain DEXs and centralized exchanges (CEXs)~\cite{hagele2024centralized}. For instance, a CEX's limit order book can be approximated as a new edge with a concave swap function, allowing its integration into our framework. The swap function is determined by the depth and price range of the order book. Since all swap functions are concave, the overall aggregated market model preserves concavity. Our PRIME algorithm remains effective on this aggregated market by discovering paths and using our ASGM to optimize the allocation of input. Therefore, our model is applicable to more markets.

\eat{We can model these liquidity source as a new edge within our token swap graph. Specifically, for liquidity from a CEX, the limit order book can be modeled as a series of discrete price levels, each offering a fixed exchange rate up to a certain quantity. Although the swap function of an order book is not continuously concave in the strict mathematical sense, a deep and liquid order book can be effectively approximated as exhibiting diminishing returns. Therefore, the order book can be modeled within our framework as a special case of edge. 
}

\vspace{2mm}
\section{Related Work}
\label{sec:relate}
\meta{
\spara{Graph Query Processing} Shortest path query processing is a fundamental problem in database research, with recent focus shifting from static graphs to constrained and dynamic environments. State-of-the-art solutions for Constrained Shortest Path (CSP) queries typically leverage indexing schemes based on tree decomposition or 2-hop labeling to prune the search space. For instance, Wang et al.~\cite{wang2024pcsp} proposed PCSP, which utilizes tree decomposition to efficiently answer label-constrained queries by pre-computing partial paths. Similarly, Zhang et al.~\cite{zhang2021efficient} introduced the LSD-Index, employing dominance pruning to handle complex topological constraints. For multi-criteria optimization, Liu et al.~\cite{liu2024approximate} developed approximate skyline indices to identify Pareto-optimal paths with theoretical guarantees. However, these techniques fundamentally assume weight additivity. In contrast, Token Graphs edge weights are concave swap function. This dependency violates the additive substructure required by existing indices. Furthermore, Zhang et al.~\cite{zhang2021experimental} demonstrated that heavy indexing schemes degrade significantly in high-churn environments, rendering them unsuitable for block-rate updates in DeFi.

\spara{Generalized Network Flow} Generalized Network Flow (GNF) extends standard flow problems by incorporating gain factors on edges to model loss or amplification, transforming flow conservation into a linear programming challenge~\cite{dantzig2016linear}. While early combinatorial approaches Truemper~\cite{truemper1973optimal} utilized logarithmic transformations to map multiplicative gains to additive costs, recent theoretical breakthroughs by V\'egh~\cite{vegh2014strongly} achieved a strongly polynomial time algorithm using continuous scaling techniques. In the context of non-linear extensions, Shigeno~\cite{shigeno2006maximum} addressed minimum cost flows with concave gains, and Angeris et al.~\cite{angeris2021analysis,angeris2022constant} recently formulated routing in Constant Function Market Makers as a convex optimization problem. However, these generalized flow algorithms predominantly rely on computationally intensive solvers, which are infeasible for real-time transaction routing on massive graphs.

}

\eat{
\spara{Blockchain Transactions Optimization} The growing adoption of decentralized exchanges (DEXs) and automated market makers (AMMs) has intensified the research on blockchain transaction. For instance, Daian et al.~\cite{daian2020flash} first identified MEV, revealing how attackers exploit transaction reordering within blocks for profit. Subsequent work by Zhou et al.~\cite{zhou2021high,qin2022quantifying} quantified the prevalence of this phenomenon on the Ethereum blockchain. Wang et al.~\cite{wang2024private} analyzed how private transactions can generate additional revenue for block builders. Parallel efforts target transaction cost optimization. For instance, Kong et al.~\cite{kong2022characterizing} identified and characterized gas-inefficient patterns in smart contracts, while Albert et al.~\cite{albert2020synthesis,albert2020gasol} developed static analysis tools to automatically detect and suggest fixes for such costly patterns. Our work improves existing research on the token graph routing and is orthogonal to other optimization methods. These optimization methods could be applied together to provide a better experience for users using blockchains.

}

\eat{
\spara{Optimizing Blockchain Transactions} 
The growing adoption of decentralized exchanges (DEXs) and automated market makers (AMMs) has intensified the need for effective token swap algorithms. \eat{A primary goal is token swap routing, which focuses on maximizing the amount of the target token received for a given input amount. Prior theoretical work has often modeled this as a convex optimization or network flow problem.} For instance, Angeris et al. ~\cite{angeris2022constant,angeris2021analysis} first proposed convex optimization method to analyze routing and arbitrage problems in token swap graph. Diamandis et al.~\cite{diamandis2023efficient,diamandis2024convex} further proposed an efficient decomposition algorithm and analyzed a more general convex network flow problem. Besides, Zhang et al.~\cite{zhang2025line,zhang2025measuring} develop a line graph-based algorithm for optimal single-path routing and propose a new metric to evaluate the efficiency of their routing algorithm.

Beyond the routing problem, other research focuses on optimization in price protection and transaction cost. \cite{daian2020flash} first identified MEV, revealing how attackers exploit transaction reordering within blocks for profit. Subsequent work by \cite{zhou2021high,qin2022quantifying} quantified the prevalence of this phenomenon on the Ethereum blockchain. \cite{wang2024private} analyzed how private transactions can generate additional revenue for block builders. Parallel efforts target transaction cost optimization. For instance, \cite{kong2022characterizing} identified and characterized gas-inefficient patterns in smart contracts, while Albert et al.~\cite{albert2020synthesis,albert2020gasol} developed static analysis tools to automatically detect and suggest fixes for such costly patterns. Our work improves existing research on the token graph routing and is orthogonal to other optimization methods. These optimization methods could be applied together to provide a better experience for users using blockchains.

\spara{Interconnecting Blockchain Ecosystems} The rapid growth of DeFi has fueled both the demand for cross-chain interoperability and research into on-chain economics. In response, several studies have explored the mechanisms enabling such interoperability. For instance, De et al.~\cite{de2150crocrpc} proposed CroCRPC, a framework for cross-chain remote procedure calls, enabling dApps to interact across chains. Similarly, Tao et al.~\cite{tao2023sharding} investigate sharding across heterogeneous blockchains, outlining the complex topological environments where routing must operate. Wang et al.~\cite{wang2024v2fs} proposed V2FS, an infrastructure for verifiable state queries in multi-chain environments. On the economic side, Demosthenous et al.~\cite{demosthenous2150chain} analyzed the use of on-chain data for cryptocurrency market forecasting, demonstrating how transaction data reflects market dynamics pertinent to DeFi. Huang et al.~\cite{huang2021rich,huang2025last} analyze the fairness of blockchain economic incentives, which can impact network fairness and  are central to user experience in DeFi. Our routing algorithm can be combined with the above interoperability frameworks and incentive analyses to enable a interconnected multi-chain DeFi market. 
}

\section{Conclusion}
\label{sec:conclude}
We study the token graph routing problem in the decentralized blockchain finance market, focusing on maximizing user output for a given input amount. We introduce a practical modeling framework grounded in Swap Paths, aligning closely with on-chain execution realities. Our proposed algorithm, PRIME, utilizes a two-stage approach: efficient path discovery via a Breadth-First Search Pruning method and optimal allocation using our novel Adaptive Sign Gradient Method (ASGM).  Our theoretical analysis shows that ASGM converges to an optimal allocation with a guaranteed linear rate. Both theoretical proofs and extensive experiments using real-world Ethereum data are carried out to demonstrate PRIME’s effectiveness. Our results show that PRIME consistently outperforms the widely-used Uniswap Smart Order Router benchmark, achieving better exchange prices, especially for larger trades and significantly faster execution times across diverse market periods.

\section*{Acknowledgment}
Jing Tang's work is partially supported by the National Natural Science Foundation of China (NSFC) under Grant No.\ 62402410 and by Guangdong Provincial Project (No.\ 2023QN10X025).

\bibliographystyle{IEEEtran}
\bibliography{sample}

\appendices

\section{Convergence Rate Analysis}
\label{app:convergence}

This appendix provides a detailed analysis of the convergence rate of the Adaptive Sign Gradient Method (ASGM) for the Token Allocation Problem. The main text states Theorem~\ref{thm:convergence} and gives a brief proof sketch; here we expand the derivation for the technical report.

\subsection{Problem Setup and Notation}

Recall the allocation problem:
\begin{equation}
\begin{split}
\max_{W} \, J(W) &= \sum_{p \in \mathcal{P}} f_p(W_p \cdot x), \\
\text{s.t.} \quad &\sum_{p \in \mathcal{P}} W_p = 1,\quad W_p \geq 0.
\end{split}
\end{equation}
The objective $J$ is strongly concave over the simplex $\Omega$. We denote the optimal allocation by $W^*$ and the optimality gap at iteration $t$ by $D(W_t) = J(W^*) - J(W_t)$.

\subsection{Proof of Theorem~\ref{thm:convergence}}
\label{app:pgm_convergence}

To derive the ASGM convergence properties, we analyze how the optimality gap $\delta_t = J^* - J(W_t)$ decreases during every iteration. Here, $J^*$ is the optimal value of the objective function $J(W)$, $W^*$ is the unique optimal allocation in the constraint set $\Omega = \{W | \sum_p W_p = 1, W_p \geq 0\}$, and $W_t$ is the allocation at iteration $t$.

We first establish several key lemmas that will be used in the main proof.

\begin{lemma}[Strong Concavity of Path Swap Functions]
\label{lem:path_strong_concavity}
Given that each edge swap function $f_e$ is strictly and strongly concave on a bounded domain, the path swap function $f_p$, defined as the composition of hop swap functions, is also strongly concave.
\end{lemma}

\begin{proof}
By Definition~\ref{def: swap function}, each $f_e$ is strictly concave and monotonically increasing.

\textbf{Weighted Sums:} Each hop function is a weighted sum of strongly concave functions, which preserves strong concavity.

\textbf{Composition:} The path function $f_p$ is the composition of these hop functions. Since the composition of monotonically increasing concave functions is concave, and strong concavity is maintained over the finite trade range, the overall objective function $J(W) = \sum_{p \in \mathcal{P}} f_p(W_p \cdot x)$ is strongly concave.
\end{proof}

\begin{lemma}[Lipschitz Gradient Property]
\label{lem:l_property}
Given a concave objective function $J$ with $L$-Lipschitz gradients, the following inequality holds:
\begin{equation}
\begin{split}
J(W_t + \Delta \cdot \text{SIGN}) &\geq J(W_t) + \Delta \cdot \nabla J(W_t)^\top \text{SIGN} \\
&\quad - \frac{L}{2} \Delta^2 \|\text{SIGN}\|^2,
\end{split}
\end{equation}
where $\Delta$ denotes the step size and $\text{SIGN}$ represents the direction vector identifying the rebalancing direction between the most and least efficient paths.
\end{lemma}

\begin{proof}
This lemma is a direct consequence of the $L$-Lipschitz gradient property for smooth functions~\cite{boyd2004convex}. The Lipschitz constant $L$ exists because the swap functions are defined on bounded domains, ensuring bounded second derivatives.
\end{proof}

\begin{lemma}[Armijo Condition]
\label{lem:armijo}
The step size $\Delta$ satisfies the Armijo condition if:
\begin{equation}
J(W_t + \Delta \cdot \text{SIGN}) \geq J(W_t) + \alpha \Delta \cdot \nabla J(W_t)^\top \text{SIGN},
\end{equation}
where $\alpha \in (0,1)$ is a constant.
\end{lemma}

\begin{proof}
This follows directly from the backtracking line search procedure in Algorithm~\ref{alg:asg}, which iteratively reduces $\Delta$ by a factor $\beta$ until the Armijo condition is satisfied.
\end{proof}

\begin{lemma}[Strong Concavity and Simplex Geometry Bound]
\label{lem:strong_concavity_bound}
For the $\mu$-strongly concave function $J$ optimized over the simplex domain $\Omega$, the rebalancing direction $\text{SIGN}$ satisfies the following linear and quadratic relations:
\begin{align}
\nabla J(W_t)^\top \text{SIGN} &\geq c \|\text{SIGN}\| \left( \nabla J(W_t)^\top \hat{e}_t \right), \\
(\nabla J(W_t)^\top \text{SIGN})^2 &\geq 2 \mu c^2 \|\text{SIGN}\|^2 \cdot \delta_t,
\end{align}
where $\delta_t = J^* - J(W_t)$ is the optimality gap, $\hat{e}_t = (W^* - W_t)/\|W^* - W_t\|$ is the normalized optimal direction, and $c = PWidth(\Omega) > 0$ is the constant pyramidal width of the simplex domain.
\end{lemma}

\begin{proof}
We begin by establishing the alignment between the algorithmic search direction and the optimal face of the domain. In the Adaptive Sign Gradient Method (ASGM), the vector $\text{SIGN}$ reallocates flow coordinates exclusively between the paths with the highest and lowest marginal prices, which structurally mimics a Pairwise Frank-Wolfe update step on the simplex constraint $\Omega$~\cite{lacoste2015global}. Leveraging the geometric properties of polytope boundaries analyzed by Lacoste-Julien and Jaggi~\cite{lacoste2015global}, the directional derivative along this pairwise rebalancing direction is bounded by the normalized optimal direction $\hat{e}_t$ as follows:
\begin{equation}
\label{eq:pfwidth_bound}
\nabla J(W_t)^\top \text{SIGN} \geq c \|\text{SIGN}\| \left( \nabla J(W_t)^\top \hat{e}_t \right).
\end{equation}

By the $\mu$-strong concavity of $J(W)$, the value at the global maximum $W^*$ satisfies $\delta_t \le \nabla J(W_t)^\top (W^* - W_t) - \frac{\mu}{2}\|W^* - W_t\|^2$. Letting $W^* - W_t = x \hat{e}_t$ with $x = \|W^* - W_t\|$, maximizing the right-hand side over $x \geq 0$ yields:

\begin{equation}
\label{eq:quadratic_max}
\delta_t \leq \max_{x \geq 0} \left\{ x (\nabla J(W_t)^\top \hat{e}_t) - \frac{\mu}{2} x^2 \right\} = \frac{(\nabla J(W_t)^\top \hat{e}_t)^2}{2\mu}.
\end{equation}

Equation~\eqref{eq:quadratic_max} implies $\nabla J(W_t)^\top \hat{e}_t \geq \sqrt{2\mu \delta_t}$. Substituting this lower bound into Equation~\eqref{eq:pfwidth_bound} gives $\nabla J(W_t)^\top \text{SIGN} \geq c \|\text{SIGN}\| \sqrt{2\mu \delta_t}$. Finally, squaring both sides yields the desired inequality $(\nabla J(W_t)^\top \text{SIGN})^2 \geq 2 \mu c^2 \|\text{SIGN}\|^2 \cdot \delta_t$, which completes the proof.

\end{proof}

\begin{lemma}[Step Size Lower Bound]
\label{lem:low_bound}
At iteration $t$, the Armijo-selected step size $\Delta_t$ satisfies
\begin{equation}
\begin{split}
\Delta_t &\geq \min\biggl\{\Delta_0,\;
\frac{2(1-\alpha)}{L} \cdot \frac{\nabla J(W_t)^\top \text{SIGN}}{\|\text{SIGN}\|^2}\biggr\},
\end{split}
\end{equation}
where $\Delta_0$ is the initial step size.
\end{lemma}

\begin{proof}
From Lemma~\ref{lem:l_property} and the Armijo condition (Lemma~\ref{lem:armijo}), we have:
\begin{equation}
\begin{aligned}
J(W_t) + \alpha \Delta \cdot \nabla J(W_t)^\top \text{SIGN}
&\leq J(W_t + \Delta \cdot \text{SIGN}) \\
&\leq J(W_t) + \Delta \cdot \nabla J(W_t)^\top \text{SIGN} \nonumber \\
&\quad - \frac{L}{2} \Delta^2 \|\text{SIGN}\|^2.
\end{aligned}
\end{equation}
Rearranging, we obtain
\begin{equation}
\Delta \geq \frac{2(1-\alpha)}{L} \cdot \frac{\nabla J(W_t)^\top \text{SIGN}}{\|\text{SIGN}\|^2}.
\end{equation} Since the backtracking procedure starts from $\Delta_0$ and reduces by $\beta$ until the condition is met, the final step size is at least the minimum of $\Delta_0$ and this lower bound.
\end{proof}

\subsubsection{Main Proof: Linear Convergence Rate}

Using the Armijo condition (Lemma~\ref{lem:armijo}), the improvement in the objective value is:
\begin{equation}
\begin{split}
J(W_{t+1}) - J(W_t) &= J(W_t + \Delta_t \cdot \text{SIGN}) - J(W_t) \\
&\geq \alpha \Delta_t \cdot \nabla J(W_t)^\top \text{SIGN}.
\end{split}
\end{equation}

Substituting the lower bound for $\Delta_t$ from Lemma~\ref{lem:low_bound}:
\begin{equation}
\begin{split}
J(W_{t+1}) - J(W_t) &\geq \alpha \cdot \min\biggl\{\Delta_0,\\
&\qquad \frac{2(1-\alpha)}{L} \cdot \frac{\nabla J(W_t)^\top \text{SIGN}}{\|\text{SIGN}\|^2}\biggr\} \\
&\qquad \cdot \nabla J(W_t)^\top \text{SIGN}.
\end{split}
\end{equation}

\begin{figure*}[t] 
  \centering
  \begin{subfigure}[b]{0.49\linewidth} 
    \centering
    \includegraphics[width=\linewidth, height=4.5cm, keepaspectratio]{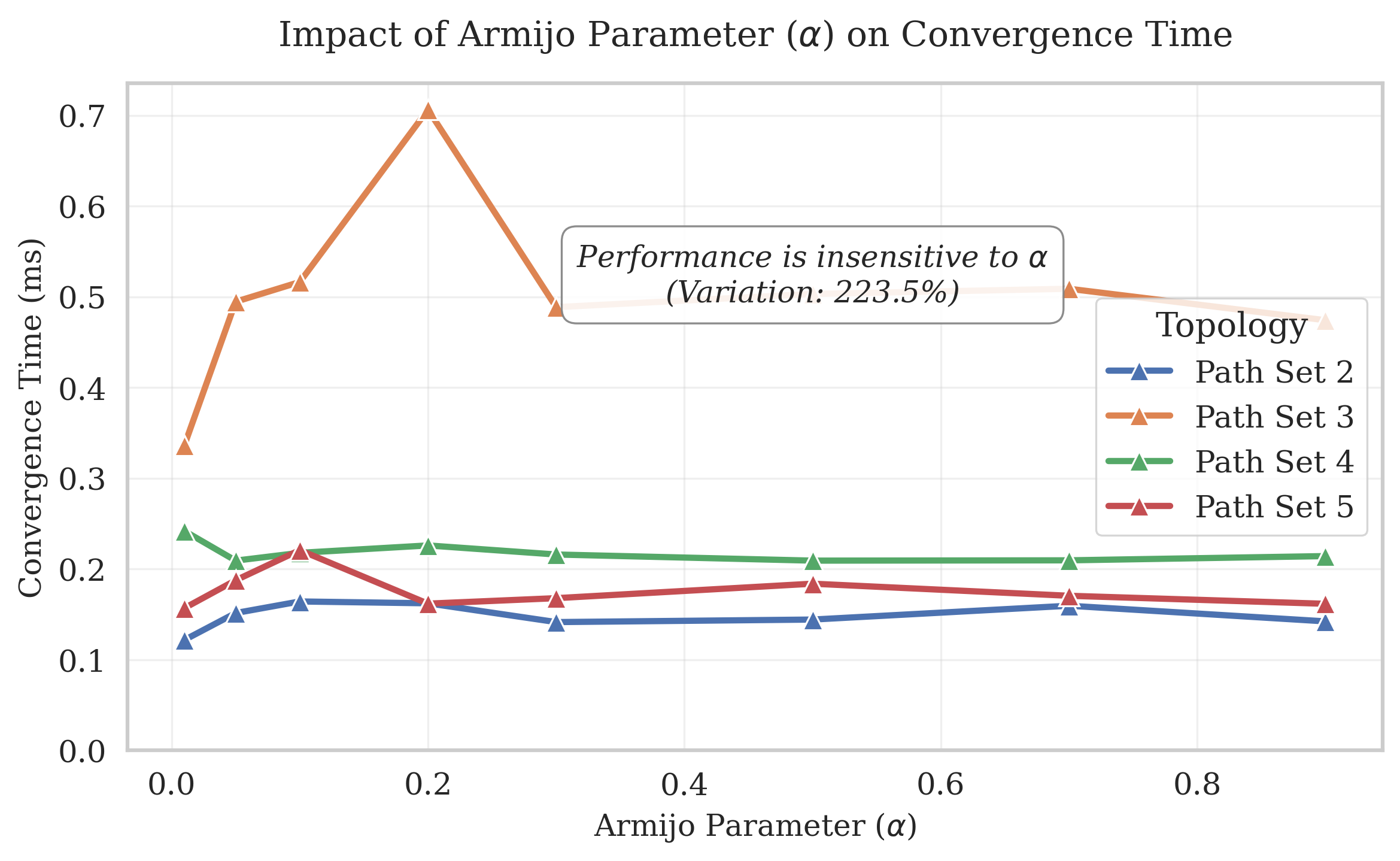}
    \caption{Robustness to Armijo parameter $\alpha$}
    \label{fig:alpha_sensitivity}
  \end{subfigure}
  \hfill
  \begin{subfigure}[b]{0.49\linewidth}
    \centering
    \includegraphics[width=\linewidth, height=4.5cm, keepaspectratio]{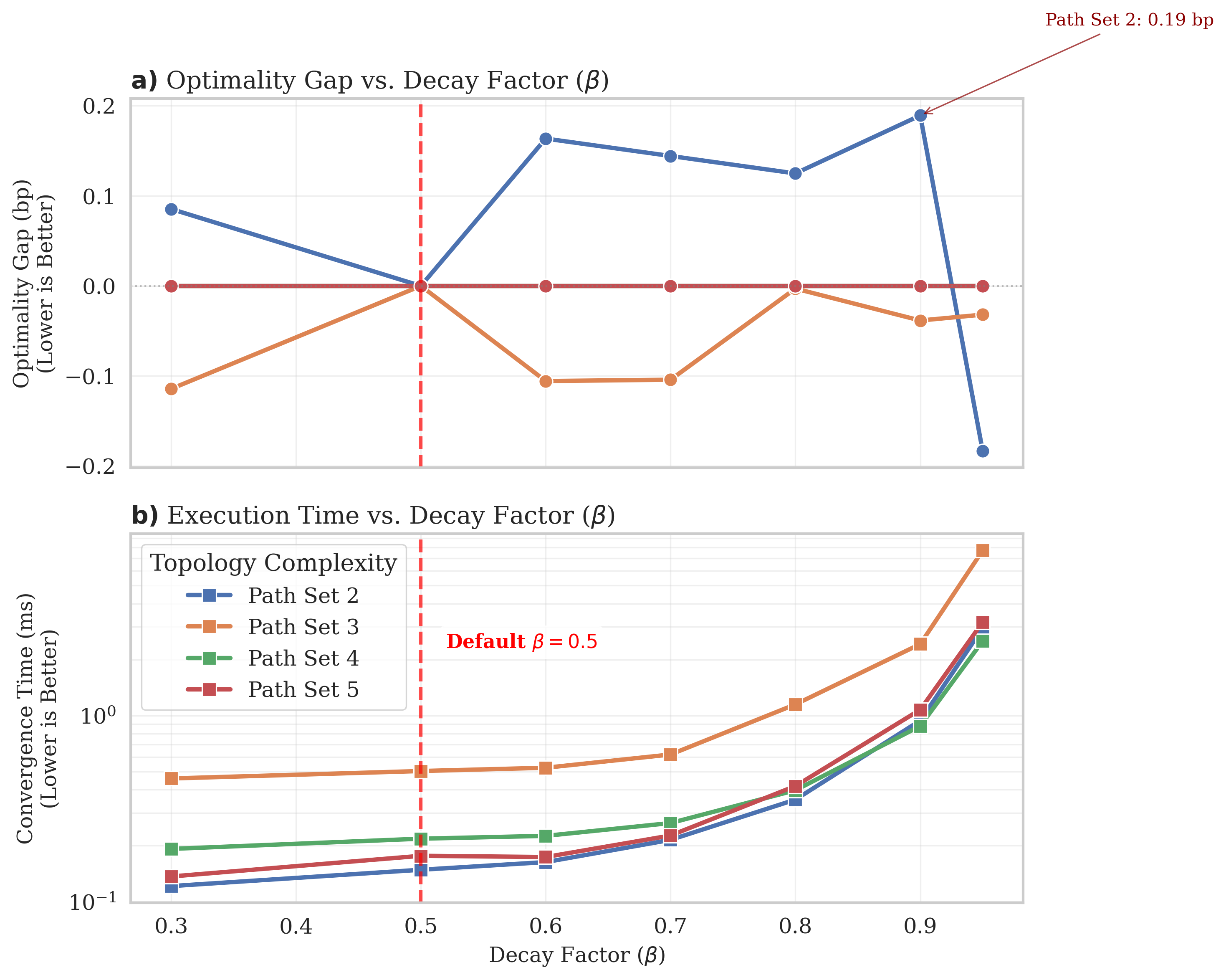}
    \caption{Trade-off analysis of decay factor $\beta$}
    \label{fig:beta_ablation}
  \end{subfigure}
  \caption{Hyperparameter Ablation Study. (a) Execution time remains flat across $\alpha \in [0.01, 0.9]$, indicating the algorithm is insensitive to the acceptance threshold. (b) While higher $\beta$ marginally reduces the optimality gap, it causes an exponential increase in convergence time. The chosen $\beta=0.5$ strikes the optimal balance.}
  \label{fig:ablation_study}
\end{figure*}

As the algorithm approaches the optimum (Phase 2), $\Delta_t$ is determined by the gradient term. Applying the bound from Lemma~\ref{lem:strong_concavity_bound}:
\begin{equation}
\begin{aligned}
J(W_{t+1}) - J(W_t) &\geq \alpha \cdot \frac{2(1-\alpha)}{L} \cdot \frac{(\nabla J(W_t)^\top \text{SIGN})^2}{\|\text{SIGN}\|^2} \\
& \geq \alpha \cdot \frac{2(1-\alpha)}{L \|\text{SIGN}\|^2} \cdot \left( 2 \mu c^2 \|\text{SIGN}\|^2 \delta_t \right) \\
& = \left(\frac{4\alpha(1-\alpha)\mu c^2}{L}\right) \delta_t.
\end{aligned}
\end{equation}

Setting $\rho = \frac{4\alpha(1-\alpha)\mu c^2}{L}$, where $c = PWidth(\Omega)$ represents the fixed geometric pyramidal width of the routing simplex domain, and using $\delta_t = J^* - J(W_t)$, we have:
\begin{equation}
\delta_{t+1} = \delta_t - (J(W_{t+1}) - J(W_t)) \leq (1 - \rho) \delta_t.
\end{equation}

This recursive relation directly implies $\delta_t \leq (1-\rho)^t \delta_0$, establishing a \textbf{linear convergence rate}. This rate ensures that PRIME converges efficiently despite extreme value differences between tokens like BTC and meme coins.

\subsection{Connection to Frank-Wolfe–Type Methods}

As in the main text, the analysis builds on the framework of Lacoste-Julien and Jaggi~\cite{lacoste2015global}. ASGM's sign-based update mechanism can be viewed as a Frank-Wolfe--style method: at each iteration we identify the two paths with the largest marginal price difference and rebalance between them. The key difference from standard Frank-Wolfe is that ASGM uses only the \emph{sign} of the gradient difference, making it robust to extreme scale disparities.

The linear convergence rate comes from the strong concavity of the objective function. In standard Frank-Wolfe, convergence is typically sublinear ($O(1/t)$) for general convex objectives. However, for strongly concave objectives on a simplex, the duality gap decreases geometrically, enabling linear convergence. The Armijo backtracking ensures that each step achieves sufficient progress, maintaining the linear rate even when the Lipschitz constant varies across the domain.


\section{Ablation Study: Minimum Learning Rate and Decay Factor}
\label{app:ablation_lr_decay}

This appendix presents the ablation study and sensitivity analysis for the ASGM hyperparameters: the \textbf{Armijo condition parameter} $\alpha$ (sufficient decrease parameter) and the \textbf{step size decay factor} $\beta$. Both control the backtracking line search in Algorithm~\ref{alg:asg}.

\subsection{Experimental Setup}

We evaluate ASGM across varying $(\alpha, \beta)$ combinations on real-world Ethereum mainnet data on Block 22918466. The evaluation used a  WETH swap across varying topological complexities represented by Path Sets 2–5 (disjoint routing paths). We perform a grid search over $\alpha \in [0.01, 0.9]$ and $\beta \in [0.3, 0.95]$, with input amount fixed at 1,000 WETH. We report optimality gap (measured in basis points, bp) and execution time (ms). Defaults in the main experiments are $\alpha = 10^{-4}$ and $\beta = 0.5$.

\subsection{Hyperparameter Sensitivity Analysis}

To validate the robustness of the proposed ASGM algorithm, we conducted a comprehensive ablation study on the Armijo condition parameter $\alpha$ and the step size decay factor $\beta$. The results, illustrated in Fig.~\ref{fig:ablation_study}, confirm that ASGM effectively balances precision and latency without requiring fine-grained parameter tuning.

\spara{Sensitivity to Armijo Parameter $\alpha$} The parameter $\alpha \in (0, 1)$ determines the strictness of the line search acceptance criterion. As shown in Fig.~\ref{fig:ablation_study}(a), varying $\alpha$ from $0.01$ to $0.9$ results in \textbf{negligible variations} in convergence time and zero variance in solution quality across all path sets. 
For instance, in Path Set 5 (complex topology), the execution time remains stable around the baseline regardless of $\alpha$. This indicates that for the convex objective function of CFMM routing, the sufficient decrease condition is easily satisfied. Consequently, the algorithm is insensitive to $\alpha$, and the standard choice of $\alpha=10^{-4}$ is robust for all tested scenarios.

\spara{Sensitivity to Decay Factor $\beta$} The decay factor $\beta \in (0, 1)$ controls the granularity of the backtracking line search, representing a trade-off between solution precision and computational speed.
\begin{itemize}
    \item \textbf{Optimality Gap (Fig.~\ref{fig:ablation_study}(b) Top):} For simpler topologies (Path Sets 1, 4, 5), ASGM finds the global optimum (Gap $\approx 0.00$ bp) across all $\beta$ values. For Path Set 2, increasing $\beta$ from $0.5$ to $0.9$ yields a marginal gain of $<0.2$ bp ($0.002\%$), which is theoretically negligible.
    \item \textbf{Execution Latency (Fig.~\ref{fig:ablation_study}(b) Bottom):} As $\beta \to 1$, the step size decreases more slowly, forcing the algorithm to perform significantly more iterations to satisfy the Armijo condition. The data shows that increasing $\beta$ from $0.5$ to $0.95$ causes latency to surge by over \textbf{20$\times$} (e.g., from $\approx 0.15$ ms to $>3.0$ ms).
\end{itemize}

The ablation study demonstrates that ASGM operates within a distinct ``sweet spot.'' The selected default $\beta=0.5$ achieves near-optimal profitability (capturing $>99.998\%$ of theoretical maximum extraction) while maintaining sub-millisecond latency essential for real-time block building. Deviating to higher $\beta$ offers diminishing returns in pricing accuracy at a prohibitive computational cost.
\balance

\end{document}